\documentclass[11pt]{article}
\usepackage{geometry}
\usepackage{amssymb}
\usepackage{mathrsfs}
\usepackage{amsmath}
\usepackage[latin1]{inputenc}
\usepackage{tikz}
\usetikzlibrary{shapes,arrows}

\newtheorem{theorem}{Theorem}[section]
\newtheorem{corollary}[theorem]{Corollary}
\newtheorem{lemma}[theorem]{Lemma}
\newtheorem{proposition}[theorem]{Proposition}
\newtheorem{claim}[theorem]{Claim}
\newtheorem{definition}[theorem]{Definition}

\def\squarebox#1{\hbox to #1{\hfill\vbox to #1{\vfill}}}
\newcommand{\qed}{\hspace*{\fill}
\vbox{\hrule\hbox{\vrule\squarebox{.667em}\vrule}\hrule}\smallskip}

\newenvironment{proof}{\noindent{\bf Proof:~~}}{\(\qed\)}

\newcommand{\ignore}[1]{}

\newcommand{\xhdr}[1]{\paragraph{{#1}.}}

\begin{document}
\title{Economic Efficiency Requires Interaction}
\author{
Shahar Dobzinski\thanks{Weizmann Institute of Science. Email: \texttt{shahar.dobzinski@weizmann.ac.il}. Incumbent of the Lilian and George Lyttle Career Development Chair. Supported in part by the I-CORE program of the planning and budgeting committee and the Israel Science Foundation 4/11 and by EU CIG grant 618128.} 
\and Noam Nisan\thanks{Hebrew University and Microsoft Research. Email: \texttt{noam@cs.huji.ac.il}.}
\and Sigal Oren\thanks{Hebrew University and Microsoft Research. Email: \texttt{sigalo@cs.huji.ac.il}. Supported by an I-CORE fellowship and by a grant from the Israel Science Foundation. Part of the work was done while at Cornell University where the author was supported in part by
NSF grant CCF-0910940 and a Microsoft Research Fellowship. }}
\date{}
\maketitle


%
%
%


\begin{abstract}
We study the necessity of interaction between individuals for obtaining approximately efficient economic allocations. 
We view this as a formalization of Hayek's classic point of view that 
focuses on the information transfer advantages that markets have relative to centralized planning. 
We study two settings: combinatorial auctions with unit demand bidders (bipartite matching) and combinatorial auctions with subadditive bidders. In both settings we prove that non-interactive protocols require exponentially larger communication costs 
than do interactive ones, even ones
that only use a modest amount of interaction.
\end{abstract}

\maketitle

\thispagestyle{empty}

\newpage
\setcounter{page}{1}





\section{Introduction}


The most basic economic question in a social system is arguably how to determine an efficient
allocation of the economy's resources. 
This challenge was named the {\em economic calculation problem} by von Mises, who
argued that markets do a better job than centralized systems can. In his classic paper \cite{H45}, Hayek claimed that the heart of the matter is the distributed nature of ``input'', i.e. that the central planner does not have the information regarding the costs and utilities of the different parties:  
\begin{quote}
\emph{knowledge of the circumstances of which we must make use never exists in concentrated or integrated form but solely as the dispersed bits of ... knowledge which all the separate individuals possess.}
\end{quote}

Hayek therefore proposes that the question of which economic system is better (market-based or centrally-planned) hinges on which of them is able 
to better transfer the information needed for economic efficiency:

\begin{quote}
\emph{which of these systems is likely to be more efficient depends ... on whether we are more likely to succeed in putting at the disposal of a single central authority all the knowledge which ... is initially dispersed among many different individuals, or in conveying to the individuals such additional knowledge as they need in order to enable them to fit their plans with those of others.  }
\end{quote}

%

When Hurwicz \cite{Hur73} formalized a notion of protocols for transfer of information in an economy, the basic examples were a Walrasian-Tatonnement protocol modeling a free market and a 
command process that entails full report of information to a centralized planner.  He noted that:  

\begin{quote}
\emph{The language of the command process is much larger than that of the Walrasian process.  We must remember, however, that the pure command process is finished after only two exchanges of information
while the tatonnement may go on for a long time.}
\end{quote}

This paper follows Hurwicz's approach of formalizing Hayek's question by technically studying the amount of 
information exchange needed for obtaining
economic efficiency.  We consider the main distinction -- in terms of information transfer -- between a centralized system 
and a distributed market-based one to be that of interaction: in a centralized system all individuals send information 
to the central planner
who must then determine an efficient allocation, while market based systems are by nature interactive.\footnote{
Interactive versions of centralized planning such as the ``market socialism'' models of Lange and Lerner \cite{LL38} have also been suggested.
These can essentially simulate free markets so become 
indistinguishable from market mechanisms in terms of their informational abilities.  The distinction between them and true market models is, thus, outside the scope of 
the discussion in this paper.  It should perhaps also be noted here that from a more modern point of view, the revelation principle states
that centralized systems can also simulate any incentives that a distributed one may have.}  

Our main results support Hayek's point of view.  We exhibit situations where interaction allows exponential savings in
information transfer, making the economic calculation problem 
tractable for interactive markets even
when it is intractable for a centralized planner.
We have two conceptually similar, but technically disjoint, sets of results along this line. The first set of results considers the classic simple setting of unit-demand bidders, essentially a model of matching in bipartite graphs.  The second and more complicated setting concerns combinatorial auctions with subadditive bidders. 
In both settings we show that non-interactive protocols 
have exponentially larger communication costs relative to interactive ones. 
In a complementary set of results we also show that exponential savings in communication costs can be 
realized with even limited interaction. 

In technical terms, we formalize this problem in the realm of communication complexity.\footnote{While Hurwicz and other economists employed models that allowed communicating real numbers,
we employ the standard, modern, notions from computer science (see \cite{KN06}) that count bits.  This distinction does not seem to have any conceptual significance -- see \cite{NS06} for a discussion.} 
Non-interactive systems are modeled as simultaneous communication protocols, where all agents simultaneously send messages to a central planner who must decide on the allocation based on these messages alone.  Interactive systems may use multiple rounds of communication and
we measure the amount of interactiveness of a system by the number of communication rounds. 
In both of our settings we prove lower bounds on the simultaneous communication complexity of finding an 
approximately efficient allocation as well as
exponentially smaller upper bounds for protocols that use a modest number of rounds of interaction.  

We now elaborate in more details on the two settings that we consider and describe our results. We begin with the technically simpler setting of bipartite matching.

\subsection{Bipartite Matching}

In this simple matching scenario there are $n$ players and $n$ goods.  Each player $i$ is interested in acquiring a single item from some privately known subset $S_i$ of the goods and our goal is to allocate the items to the players in a way that maximizes the number of players who get an item from their desired set. This is of course a classic problem in economics (matching among unit demand bidders) as well as in computer science (bipartite matching).  

We first consider simultaneous protocols. Each of the players is allowed to send a small amount of information, $l$ bits with $l<<n$, to the centralized planner who must then output a matching.

\vspace{0.1in}
\noindent
{\bf Theorem:} 
\begin{itemize}
\item Every deterministic simultaneous protocol where each player sends at most $n^\epsilon$ bits of communication cannot approximate the size of the maximum matching to within a factor better than $O(n^{1-\epsilon})$.
\item Any randomized simultaneous protocol where each player sends at most $n^\epsilon$ bits of communication cannot approximate the size of the maximum matching to within a factor better than $O(n^{\frac 1 2-2\epsilon})$.
\end{itemize}
\vspace{0.1in}

\noindent Both our bounds are essentially tight. For deterministic protocols, one can trivially obtain an approximation ratio of $n$ with message length $O(\log n)$: each player sends the index of one arbitrary item that he is interested in. If randomization is allowed, it is not hard to see that when each player sends the index of a random item he is interested in, we get an approximation ratio of $O(\sqrt n)$. We have therefore established a gap between randomized and deterministic protocols. We also note that the randomized lower bound can in fact be obtained from more general recent results of \cite{HBMQ13} in a stronger model. For completeness we present a significantly simpler direct proof for our setting.


On the positive side, we show that a few communication rounds suffice to get an almost efficient allocation\footnote{Formally, the algorithm works in the so called ``blackboard model'' -- see Appendix \ref{appendix-matching} for a definition.}. Our algorithm is a specific
instantiation of the well known class of ``auction algorithms''. This class of algorithms has its roots in \cite{DGS86} and has been extensively studied from an economic point of view (e.g. \cite{RS92}) as well as 
from a computational point of view (starting with \cite{B88}).

The standard ascending auction algorithm for this setting
begins by setting the price of each item to be $0$. Initially, all bidders are ``unallocated''. Then, in an arbitrary order, each unallocated player $i$ reports an index of an item that maximizes his profit in the current prices (his ``demand'').
The price of that item increases by $\epsilon$ and the item is reallocated to player $i$, and the process continues 
with another unallocated player. It is well known that if the maximum
that a player is willing to pay for an item is 1, then this process terminates after at most $\Omega(\frac n \epsilon)$ steps and it is not hard to construct examples where this is tight. We show that if in each round every unallocated player reports, simultaneously with the others, an index of a \emph{random} item that maximizes his profit in the current prices ($O(\log n)$ bits of information) and each reported
item is re-allocated to an arbitrary player that reported it, then 
the process terminates in logarithmically many rounds.
We are not aware of any other scenario where
natural market dynamics provably converge 
(approximately) to an equilibrium in time that is sub-linear in the
market size.

\vspace{0.1in}
\noindent
{\bf Theorem:} 
Fix $\epsilon>0$. After $O(\frac{\log n}{\epsilon^2})$ rounds the randomized algorithm provides in expectation an $(1+\epsilon)$-approximation to the bipartite matching problem.
\vspace{0.1in}

\noindent We then quantify the tradeoff between the amount of interaction and economic efficiency. We show that for every $k \ge 1$ there is a randomized protocol that obtains an $O(n^{1/(k+1)})$-approximation in $k$ rounds, where at each round each player sends $O(\log n)$ bits of information. 

In passing we note that the communication complexity of the exact problem, i.e. of finding an exact
perfect matching, when it exists, remains a very interesting open problem.  Moreover, we believe that it
may shed some light on basic algorithmic challenges of finding a perfect matching in near-linear time as well as
deterministically in parallel.  We shortly present this direction in appendix \ref{appendix-matching}.

\subsection{Combinatorial Auctions with Subadditive Bidders}

Our second set of results concerns a setting where we are selling $m$ items 
to $n$ bidders in a combinatorial
auction.  Here each player $i$ has a valuation $v_i$ that specifies his value
for every possible subset $S$ of items. The goal is to maximize
the ``social welfare'' $\sum_i v_i(A_i)$, where $A_i$ is the set of goods that is allocated 
to player $i$.
The communication requirements in such settings have received
much attention and it is known that,
for arbitrary valuations, exponential amount of communication is required 
to achieve even $m^{1/2-\epsilon}$-approximation of the optimal welfare \cite{NS06}.
However, it is also known that if the valuations are subadditive, 
$v_i(S \cup T) \le v_i(S) + v_i(T)$, then constant factor approximations
can be achieved using only polynomial communication \cite{F06,FV06, V08, DS06, lehmann2006combinatorial}. Can
this level of approximate welfare be achieved by a direct mechanism, without
interaction?

Two recent lines of research touch on this issue.  On one hand several recent papers
show that valuations cannot be ``compressed'', even approximately, and that
any polynomial-length description of subadditive valuations (or
even the more restricted XOS valuations) must lose a factor of 
$\Theta(\sqrt{m})$ in precision \cite{BCIW12, BDFKNR12}.  
Similar, but somewhat weaker, non-approximation 
results are also known for the far more restricted subclass of ``gross-substitutes'' valuations \cite{BH11}
for which exact optimization is possible with polynomial communication.  
Thus the natural approach for a direct mechanism where each player sends a succinctly 
specified approximate version of his valuation (a ``sketch'') to the central planner cannot lead to a better 
than $O(\sqrt{m})$ approximation.  This does not, however, rule out other approaches
for non-interactive allocation, that do not require approximating the whole valuation.
Indeed we show that one can do better:

\vspace{0.1in}
\noindent
{\bf Theorem:} There exists a deterministic communication protocol 
such that each player holds a subadditive valuation and 
sends (simultaneously with the others) 
polynomially many bits of communication
to the central planner that guarantees an $\tilde O(m^{1/3})$-approximation 
to the optimal allocation.
\vspace{0.05in}

\noindent Another line of relevant research considers bidders with such valuations being put
in a game where they can only bid on each item separately \cite{christodoulou2010bayesian, BhawalkarR11,HassidimKMN11,FFGL13}.  
In such games the message of each bidder is by definition only $O(m)$ real numbers, each can be 
specified in sufficient precision with logarithmically-many bits.  Surprisingly,
it turns out that sometimes this suffices to get constant factor approximation of the 
social welfare.  Specifically one such result \cite{FFGL13} considers a 
situation where the
valuation $v_i$ of each player $i$ is drawn independently from a commonly known 
distribution $D_i$ on subadditive valuations.  In such a case, every player $i$ can 
calculate bids on the items -- requiring $O(m\log m)$ bits of communication -- based only
on his valuation $v_i$ and the distributions $D_j$ of the others (but
not their valuations $v_j$). By allocating each item
to the highest bidder we get a 2-approximation to the social welfare, in expectation
over the distribution of valuations.  This is a non-interactive protocol that
comes tantalizingly close to what we desire: all that remains is for the 2-approximation
to hold for every input rather than in expectation.  Using Yao's principle,
this would follow (at least for a randomized protocol) if we could get an 
approximation in expectation for 
{\em every} distribution
on the inputs, not just the product distribution where the valuations are
chosen independently. While the approximation guarantees 
of \cite{christodoulou2010bayesian, BhawalkarR11,HassidimKMN11,FFGL13} do not hold for 
correlated distributions on the valuations, there are 
other settings where similar approximation results do hold even for correlated 
distributions \cite{LemeST12itcs,BLNP13}.  Would this be possible here too?  

Our main technical construction proves a negative answer and
shows that interaction is essential for obtaining approximately optimal allocation among 
subadditive valuations (even for the more restricted XOS valuations): 

\vspace{0.1in}
\noindent
{\bf Theorem:} No (deterministic or randomized) protocol 
such that each player holds an XOS valuation and simultaneously with the others
sends sub-exponentially many bits of communication to 
a central planner can guarantee an $m^{\frac 1 4 -\epsilon}$-approximation.  
\vspace{0.1in}

\noindent Again, this is in contrast to interactive protocols that can achieve a factor $2$ approximation \cite{F06} (with polynomially many rounds of communication). The lower bound shows that interaction is necessary to solve the economic calculation problem in combinatorial auctions. We show that if a small amount of interaction is allowed then one can get significantly better results:

\vspace{0.1in}
\noindent{\bf Theorem:} For every $k \ge 1$ there is a randomized protocol that obtains an 
$\tilde O(k\cdot m^{1/(k+1)})$-approximation in $k$ rounds, where at each round each player sends $poly(m,n)$ bits. 
In particular, after $k=\log m$ rounds we get a poly-logarithmic approximation to the welfare.

\xhdr{Open Questions} 
In our opinion the most intriguing open question is to determine the possible approximation ratio achievable by simultaneous combinatorial auctions with submodular or even gross-substitutes players that are allowed to send $poly(m,n)$ bits. 
Our $O(m^{\frac 1 3})$-algorithm is clearly applicable for both settings. 
We do know that exactly solving the problem for gross-substitutes valuations requires exponential communication 
(see Section \ref{app-gs}),
although when interaction is allowed polynomial communication suffices. 

Another natural open question is to prove lower bounds on the approximation ratio achievable by $k$-round protocols. Our bounds only hold for $k=1$. Furthermore, how good is the approximation ratio that can be guaranteed when incentives are taken into account? Can a truthful $k$-round algorithm guarantee poly-logarithmic approximation in $O(\log n)$ rounds for XOS valuations?

For the bipartite matching setting we leave open the question of developing algorithms for weighted bipartite matching. In addition, our $k$-round algorithms are randomized; developing deterministic $k$-round algorithms even for the unweighted case is also of interest. In Appendix \ref{appendix-matching} we further discuss more open questions related to the communication complexity of bipartite matching and its relation to the computational complexity of bipartite matching. 

Finally, we studied the matching problem in the framework of simultaneous communication complexity. A fascinating future direction is to study other classic combinatorial optimization problems (e.g., minimum cut, packing and covering problems, etc.) using the lenses of simultaneous communication complexity.
%
%
%
%


\section{Preliminaries} \label{sec-prelim}

\noindent \textbf{Combinatorial Auctions.} In a combinatorial auction we have a set $N$ of players ($|N|=n)$ and a set $M$ of different items ($|M|=m$). Each player $i$ has a valuation function $v_i:2^M\rightarrow \mathbb R$. Each $v_i$ is assumed to be normalized ($v_i(\emptyset)=0)$ and non decreasing. The goal is to maximize the social welfare, that is, to find an allocation of the items to players $(A_1,\ldots, A_n)$ that maximizes the welfare: $\Sigma_iv_i(A_i)$. 
A valuation function is \emph{subadditive} if for every two bundles $S$ and $T$, $v(S)+v(T)\geq v(S\cup T)$. A valuation $v$ is \emph{additive} if for every bundle $S$ we have that $v(S)=\Sigma_{j\in S}v(\{j\})$. A valuation $v$ is \emph{XOS} if there exist additive valuations $a_1,\ldots, a_t$ such that for every bundle $S$, $v(S)=\max_r a_r(S)$. Each $a_r$ is a \emph{clause} of $v$. If $a\in\arg\max_r a_r(S)$ then $a$ is a \emph{maximizing clause} of $S$.

\vspace{0.1in} \noindent \textbf{Matching.} Here the goal is to find a maximum matching in an undirected bipartite graph $G=(V_1,V_2,E)$, $|V_1|=|V_2|=n$. Each player $i$ corresponds to vertex $i\in V_1$ and is only aware of edges of the form $(i,j)$ ($j\in V_2$ since the graph is bipartite). The neighbor set of $i$ is $S_i=\{j|(i,j)\in E\}$. The goal is to maximize the number of matched pairs. When convenient we will refer to vertices on the left as unit demand bidders and the vertices on the right as goods. Under this interpretation the neighbor set of player $i$ is simply the set of goods that he is interested in.

\vspace{0.1in} \noindent \textbf{Chernoff Bounds.}
Let $X$ be a random variable with expectation $\mu$. Then, for any $\delta>0$: $P(X>(1+\delta)\mu) < (\dfrac{e^\delta}{(1+\delta)^{1+\delta}})^\mu = (\dfrac{e^\delta}{e^{(1+\delta)\ln(1+\delta)}})^\mu = (\dfrac{1}{e^{(1+\delta)\ln(1+\delta)-\delta}} )^{\mu}$. For $\delta>e^2$ we can loosely bound this expression by:  $(\dfrac{1}{e^\delta} )^{\mu}$.


\newcommand{\nset}{S}
\newcommand{\nseti}[1]{\nset_{#1}}

\section{Lower Bounds for Bipartite Matching}\label{sec-matching-lb}

In this section we state
lower bounds on the power of algorithms for bipartite matching. 
The first one deals with the power of deterministic algorithms (proof in Subsection \ref{subsec-matching-hardness}):

\begin{theorem}[lower bound for deterministic algorithms]\label{thm-matching}
The approximation ratio of any deterministic simultaneous algorithm for matching that uses at most $l$ bits per player is no better than $\dfrac{n}{8l + 4\log(n)}$. In particular, for any fixed $\epsilon>0$ and $l=n^\epsilon$ the approximation ratio is $\Omega(n^{1-\epsilon})$.
\end{theorem}

%
The second theorem gives a lower bound on the power of randomized algorithms (proof in Subsection \ref{subsec-matching-rand}):

\begin{theorem}[lower bound for randomized algorithms]\label{theorem-matching-rand-lb}
Fix $\epsilon>0$. The approximation ratio of every algorithm in which each player sends a message of size $l\leq n^{\frac 1 2 - \alpha-\frac \epsilon 2}$ is at most $n^{\alpha}$, for every $\alpha\leq \frac 1 2 -\epsilon$.
\end{theorem}

The next proposition 
shows that both lower bounds are essentially tight. In particular, this implies a proven gap between the power of deterministic and randomized algorithms.
\begin{proposition}\label{prop-matching-trivial} \
\begin{enumerate}
\item There exists a deterministic simultaneous algorithm that uses $l$ bits per player and provides an approximation ratio of $\max(2,\frac {n\log n} l)$.
\item There exists a simultaneous randomized algorithm that provides an expected approximation ratio of $O(\sqrt n)$.
\end{enumerate}
\end{proposition}
\begin{proof}
The randomized algorithm will be obtained as a corollary of the $k$-round algorithm of Subsection \ref{subsec-k-round-matching} ($k=1$). We now describe the deterministic algorithm.

Let $l'=\frac l{\log n}$. We consider the algorithm where each player reports the indices of some arbitrary $l'$ vertices in his neighbor set. The algorithm matches as many reported vertices as possible.

We now analyze the approximation ratio of the algorithm. Consider some optimal matching. We distinguish between two cases. The first one is when at least half of the players that are matched in the optimal solution have at most $l'$ neighbors (call this set $S$). In this case, vertices in $S$ can report their full neighbor set. The algorithm will consider in particular the matching that matches vertices in $S$ as in the optimal solution, and thus we will get a $2$ approximation.

In the second case, most of the players that are matched in the optimal solution have more than $l'$ neighbors (call this set $T$). Consider the matching that the algorithm outputs. If all vertices in $T$ are matched, then we get a $2$ approximation. Otherwise, there is a player $i\in T$ that is not matched by the algorithm. This implies that all the $l'$ vertices that he reported are already matched, and this is a lower bound to the number of matches that the algorithm outputs. The theorem follows since the optimal matching makes at most $n$ matches.
\end{proof}

\subsection{Hardness of Deterministic Matching}\label{subsec-matching-hardness}

In the proof we assume that our graphs are $w$-uniform graphs. That is, the size of neighbor set $\nseti{i}$ for every player $i$ is $w$. Notice that this assumption only makes our hardness result stronger. For this proof, we fix a specific simultaneous algorithm and analyze its properties. 

We begin by considering a specific random construction of graphs that we name $w$-random. In a $w$-random graph we choose $|U|=w$ vertices from $V_2$ uniformly at random and let the neighbor set of each one of the players be $\nseti{i}=U$. Let $(a_1,\ldots,a_n)$ be the output of the algorithm on this instance. As the optimal matching includes exactly $w$ matched players it is clear that the solution outputted by the algorithm matches at most $w$ players. The crux of the proof is constructing a ``fooling instance'', where all players send the same messages and hence the algorithm cannot distinguish the fooling instance from the original instance and outputs the same allocation. We will construct this fooling instance so that on one hand for almost every player $i$, $a_i$ is not in the neighbour set of player $i$ (this will be true to approximately $n-2w$ players). This implies that the value of the matching that the algorithm outputs in the fooling instance is low. On the other hand, the size of the optimal matching in the fooling instance will be $\alpha=\Theta(n)$. 


Let $g_i(V)$ denote the message that player $i$ sends when his neighbor set is $V$. Let $G_i(m) = \{V|g_i(V)=m\}$ and $G_i(V)=G_i(g_i(V))$. The main challenge in the proof is constructing a fooling set with a large optimal matching. The next definition provides us the machinery required for proving this:

\begin{definition}
For a player $i$ and neighbor set $\nseti{i}$, $\nseti{i}$ is \emph{$\alpha$-unsafe} for a vertex $k\notin \nseti{i}$ if $|\cup_{U\in G_i(\nseti{i}), k \notin U} U |\geq \alpha$. Otherwise, we say that $S_i$ is $\alpha$-safe for $k$.
\end{definition}

For a $w$-random graph we say that player $i$ is $\alpha$-adaptable if the neighbor set $S_i$ is $\alpha$-unsafe for $a_i$. We will first formalize the discussion above by constructing a fooling instance given that there are $|P|$ $\alpha$-adaptable players. Later, we will show that there exists an instance in which $|P|\approx n-2w$ players are $\alpha$-adaptable.

\begin{lemma}
Consider a $w$-random graph in which a set $P$ of the players are $\alpha$-adaptable. Then, there exists a fooling instance for which at least $\min(|P|,\alpha)$ vertices can be matched, but the algorithm returns a solution with at most $w$ matched vertices.
\end{lemma}
\begin{proof}
The fooling instance is constructed as follows: for each player $i\in P$, let $\nseti{i}'$ be some neighbor set such that $a_i\notin \nseti{i}'$ and $\nseti{i}'\in G_i(\nseti{i})$. Such $\nseti{i}'$ exists since by assumption player $i$ is an $\alpha$-adaptable player which implies that $a_i$ is $\alpha$-unsafe for $\nseti{i}$. For any such choice of $\nseti{i}'$ we have that players in $P$ are still not matched by the algorithm. The number of matchings in the fooling instance is therefore at most $w$. Notice that this argument holds for \emph{any} set of $\nseti{i}'$'s chosen as above.

To guarantee that at least $\min(|P|,\alpha)$ vertices can be matched in the new instance, we have to be more careful in our choice of the $\nseti{i}'$'s. We say that player $i$ with neighbor set $\nseti{i}$ is \emph{interested in} vertex $j\in V_2$ if there exists a neighbor set $\nseti{i}'$ where $a_i\notin \nseti{i}'$, $\nseti{i}' \in G_i(\nseti{i})$, and $j\in \nseti{i}'$. Since for every $i\in P$ we know that $a_i$ is $\alpha$-unsafe for $\nseti{i}$, there exists at least $\alpha$ vertices which player $i$ is interested in. Thus by Hall's marriage theorem there exists a matching of at least $\alpha$ vertices among the players in $P$ where each player is matched to a vertex that he is interested in. This implies that there exists a set of neighbor set $\nseti{i}'$ as above where at least $\min(|P|,\alpha)$ vertices can be matched in the fooling instance.
\end{proof}

We now show that there exists a $w$-random graph in which the number of $\alpha$-adaptable players is large:

\begin{lemma}\label{lemma-matching-adaptable}
Let $p=\dfrac{2^l\cdot n{\alpha \choose w}}{{n \choose w} }$. There exists a $w$-random graph in which at least $(1-p)\cdot n-w$ of the players are $\alpha$-adaptable.
\end{lemma}
\begin{proof}
We first show that for each player $i$ the number of neighbor sets that are $\alpha$-safe for at least one vertex is small. 
\begin{claim}\label{clm-matching-random-v}
For each player $i$ there are at most $2^l \cdot n{\alpha \choose w}$ possible neighbor sets (of size $w$) that are $\alpha$-safe for at least one vertex. 
\end{claim}
\begin{proof}
Consider a message $m$ and a vertex $k\in V_2$. Observe that by definition for every set $S$ which is $\alpha$-safe for $k$ we have that $S \subseteq \cup_{U\in G_i(\nseti{i}), k \notin U} U$ and that $|\cup_{U\in G_i(\nseti{i}), k \notin U} U|\leq \alpha$. This immediately implies that there are at most ${\alpha \choose w}$ neighbor sets in $G_i(m)$ that are $\alpha$-safe for $k$. This is true for each of the $n$ vertices in $V_2$ and hence there are at most $n{\alpha \choose w}$ neighbor sets in $G_i(m)$ that are $\alpha$-safe for at least one vertex. The proof is completed by observing that there are at most $2^l$ different messages.
\end{proof}

We are now ready to show that there exists a $w$-random graph with the required number of $\alpha$-adaptable players. Observe that in a $w$-random graph, for each player $i$ the probability that $\nseti{i}$ is $\alpha$-safe for some vertex $k \notin \nseti{i}$ is at most $\dfrac{2^l\cdot n{\alpha \choose w}}{{n \choose w} } =p$. This is simply because by construction the neighbor set $\nseti{i}$ of player $i$ is chosen uniformly at random from all possible neighbor sets (the neighbor sets of any two players are indeed correlated). 

We now show that there exists a $w$-random graph in which there is a set $P'$, $|P'|\geq (1-p) \cdot n$ where for each $i\in P'$ we have that $\nseti{i}$ is $\alpha$-unsafe for every $k \notin \nseti{i}$. To see why this is the case, for each player $i$, let $n_i$ be a random variable that gets a value of $1$ if $\nseti{i}$ is $\alpha$-unsafe for every $k \notin \nseti{i}$ and a value of $0$ otherwise. Let $n'=\Sigma_i n_i$. By the first part, for any player $i$, $E[n_{i}]\geq 1-p$. Using linearity of expectation, $E[n']\geq n\cdot (1-p)$. The claim follows since there must be at least one instance $I$ where $n'\geq E[n']$.

To conclude the proof, observe that in the instance $I$ any player $i\in P$ for which $a_i \notin S_i$ is an $\alpha$-adaptable player. Since there are at most $w$ players for which $a_i \in S_i$ we have that at least $|P'|-w= (1-p)\cdot n-w$ players are $\alpha$-adaptable, as required.
\end{proof}

Finally, we compute the values of our parameters and show that the theorem indeed holds.

\begin{lemma}
For $\alpha =\frac{n}{4}$ and $w=2l + \log(n)$, the approximation ratio of a simultaneous algorithm in which each message length is at most $l$ bits is $\dfrac{n}{8l+ 4\log(n)}$.
\end{lemma}
\begin{proof}
We first compute a lower bound on the number of $\alpha$-adaptable players. By Lemma \ref{lemma-matching-adaptable}, we have that the number of $\alpha$-adaptable players is a least $(1-p) \cdot n -w$. By Claim \ref{matching-claim-bound1} below we have that $(1-p)\geq (1-\left( \frac{1}{4} \right)^l)$, thus $(1-p) \cdot n -w>\alpha$. Therefore the approximation ratio of the algorithm is at most $\dfrac{\alpha}{w} = \dfrac{\frac{n}{4}}{2l+ \log(n)} = \dfrac{n}{8l+ 4\log(n)}  $.
\end{proof}

\begin{claim}\label{matching-claim-bound1}
For $\alpha =\frac{n}{4}$ and $w=2l + \log(n)$: $p\leq\frac {2^l \cdot n{\alpha \choose w}} {{n \choose w}}  \leq \left( \frac{1}{4} \right)^{ l}$.
\end{claim}
\begin{proof}
 \begin{align*}
\frac {2^l \cdot n{\alpha \choose w}} {{n \choose w}} = 2^l n \cdot \frac{\alpha! \cdot (n-w)!}{(\alpha-w)! \cdot n!}= 2^l n \cdot \frac{\prod_{i=1}^w (\alpha -w+i) }{\prod_{i=1}^w (n -w+i)} \leq 2^l n \cdot \left( \frac{\alpha}{n} \right)^w.
\end{align*}
By plugging in the values of $\alpha$ and $w$ we get:
\begin{align*}
 2^l n \cdot \left( \frac{\alpha}{n} \right)^w \leq  2^{l} n \cdot \left( \frac{1}{4} \right)^{2l+\log(n)} \leq n \left( \frac{1}{4} \right)^{l+ \log(n)} \leq \left( \frac{1}{4} \right)^{ l}.
\end{align*}
 \end{proof}


\subsection{Hardness of Randomized Matching}\label{subsec-matching-rand}

We consider a bipartite graph $(V_1,V_2,E)$ with $n$ vertices in each side. As usual, the left-side vertices are the players. We prove a lower bound for randomized algorithms in this setting. By Yao's principle, it is enough to prove a lower bound on the power of deterministic mechanisms on some distribution. 

The hard distribution on which we will prove the lower bound is the following:

\begin{enumerate}  
\item The size of the neighbor set $S_i$ of each player $i$ is exactly $k+1$, where $k=n^{\frac 1 2}$.
\item The neighbor sets $S_i$ are chosen, in a correlated way, as follows: a set $T$ of size exactly $2k$ is chosen uniformly at random, and each $S_i$
is obtained by independently taking a random subset of size exactly $k$ of $T$ and another single random element from $T^c$ (the complement of $T$). 
\item The players do not know $T$ nor do they know which of their elements is the one from $T^c$.
\end{enumerate}

We will prove the lower bound by reducing the matching problem to the following two player problem.

\subsubsection{A 2-Player Problem: The Hidden Item}

In the \emph{hidden item problem} there are two players (Alice and Bob) and $n$ items. In this problem Alice holds a subset $T$ of the items of size exactly $2k$ ($k=n^{\frac 1 2}$). Bob holds a set $S$ of size exactly $k+1$. The guarantee is that $|S\cap T|=k$. Bob sends a message to Alice of length $l$ who must output an item $x$ (based only on the message that she got and $T$). Alice and Bob win if $x\in S- T$.

We will analyze the power of deterministic mechanisms on the following distribution:

\begin{enumerate}
\item $T$ is selected uniformly at random among all subsets that consist of exactly $2k$ items.
\item $S$ is selected in a correlated way by taking a random subset of size $k$ from $T$ and a random extra element from $T^c$ (without knowing which is which). 
\end{enumerate}

\begin{lemma}\label{lemma-hidden-item-hardness}
If inputs are drawn from the above distribution, in any deterministic algorithm the probability that Alice and Bob win is at most $O(n^{-\alpha})$, for any $\alpha$ and $l$ such that $n^{\alpha} \cdot l = o(n^{\frac 1 2})$.
\end{lemma}

The lemma will be proved in Subsection \ref{section-hidden-item}. We first show why the lemma implies Theorem \ref {theorem-matching-rand-lb}.

\begin{lemma} \label{lemma-matching-rand}
Let $\alpha< \frac{1}{2} -\epsilon$. If there exists a deterministic algorithm for the hard distribution of the matching problem that provides an approximation ratio of $n^{\alpha}$ where each player sends a message of length $l$, then there exists an algorithm for the hard distribution of the hidden item problem where Bob sends a message of length $l$ and the probability of success is $\frac 1 {20n^{\alpha}}$.
\end{lemma}
\begin{proof}
Assume a deterministic algorithm for the hard matching distribution achieving an approximation ratio better than $n^{\alpha}$. Observe that the optimal social welfare is at least $\frac n {10}$ with very high probability. Thus, an expected social welfare of $\frac{n^{1-\alpha}}{10}$ is required for achieving this approximation ratio. Clearly at most $2k=2n^{\frac 1 2}<n^{\frac 1 2 + \epsilon}/20 <n^{1-\alpha}/20$ of this expected social welfare comes from items in $T$ (for big enough $n$). Thus, the expected social welfare, obtained just from items outside $T$ is at least $\frac {n^{1-\alpha}} {20}$.  

This implies that there exists some player, without loss of generality player $1$, whose expected value, not including any item in $T$, 
is at least $\frac{1}{20 n^{\alpha}}$.  We will use this protocol to construct the two-player protocol, by Bob simulating player $1$ and Alice simulating all the other players combined.

When Alice and Bob get their inputs $S$ and $T$ for the hidden item problem, Alice uses $T$ to choose at random $S_2 ,..., S_n$ as to fit the distribution of the $n$-player problem, 
and Bob sets $S_1=S$. Bob sends to Alice the message player $1$ sends in the $n$-player algorithm for matching.  Alice first simulates the messages of all players and then calculates the outcome $(a_1,...,a_n)$ of the $n$-player matching. Notice that whenever player $1$ in the $n$-player protocol gets utility $1$ from an item outside of $T$,  we have that $x \in S-T$. Thus Alice and Bob win with probability at least $\frac{1}{20 n^{\alpha}}$.
\end{proof}

Theorem \ref{theorem-matching-rand-lb} now follows as we have by Lemma \ref{lemma-hidden-item-hardness} that for $l \leq n^{\frac 1 2-\alpha -\frac \epsilon 2}$ and $\alpha \leq \frac{1}{2}-\epsilon$ the probability of success in the hidden item problem is $O(n^{-\alpha})$. This together with Lemma \ref{lemma-matching-rand} implies that there cannot be an algorithm for matching achieving an approximation ratio better than $n^\alpha$ using $l\leq n^{\frac 1 2-\alpha -\frac \epsilon 2}$ bits for $\alpha \leq \frac{1}{2}-\epsilon$. 

\subsubsection{Proof of Lemma \ref{lemma-hidden-item-hardness}}\label{section-hidden-item}

Assume that a protocol with a better winning probability than $O(n^{-\alpha})$ exists. We will use random-self-reducibility to get a randomized (public coin) protocol that will work for {\em all} pairs of sets $S,T$ with $|S|=k+1$, $|T|=2k$, $|S-T|=1$.  This is obtained by the two players choosing jointly a random permutation of the $n$ items and running the protocol on their permuted items.  Notice that this randomized reduction maps any original input to exactly the distribution on which we assumed the original protocol worked well, and thus now the winning probability, for any fixed input of the specified form, is at least $O(n^{-\alpha})$. We now run this protocol $O(n^{\alpha})$ times in parallel (with independent random choices) to get a situation where Bob sends
$O(n^{\alpha} \cdot l)$ bits and Alice outputs a set of size $O(n^{\alpha})$ that with high probability contains the element in $S-T$.

We will use the known hardness for the disjointness two-player communication-complexity problem:

\begin{theorem}[Razborov]
Assume that Alice holds a subset $S$ and Bob holds a subset $R$ of a universe of size $m$ where, $|S|=m/4$, $|T|=m/4$.
Distinguishing between the case that $S \cap R = \emptyset$ and the case that $|S \cap R| = 1$ requires $\Omega(m)$ randomized (multiple round, constant-error) communication.
\end{theorem}

\begin{corollary}
By using a simple padding argument, we have that for $m=4k$ and a setting where Alice holds a set $S$ of size $k+1$ of a universe of size $n$ and Bob holds a set $R$ of size $n-2k$ of the same universe, distinguishing between $S \cap R = \emptyset$ and $|S \cap R| = 1$ requires $\Omega(m)$ randomized communication.
\end{corollary}


We now show how to use an algorithm for the hidden instance problem to solve the hard problem described in the previous corollary: let $T=R^c$ and use the protocol constructed in which Alice outputs a set of size  $O(n^{\alpha})$ that with high probability contains the element in $S-T = S \cap R$, she will now send this whole list back to Bob who will report which of these elements is in $S$.  We have now achieved a 2-round protocol (Bob $\rightarrow$ Alice $\rightarrow$ Bob) that uses $O(n^{\alpha}) \cdot l$ bits of communication that finds the element in $S \cap R$ with high probability, if such an element existed.  Otherwise, such an element is not found so we have distinguished the two possibilities. 

To reach a contradiction our protocol has to use less than $\Omega(m)$ bits. Thus for a contradiction to be reached it must hold that $n^{\alpha} \cdot l = o(m)=o(k)$.


\section{Algorithms for Bipartite Matching}\label{sec-matching-alg}

We provide two algorithms that guarantee significantly better approximation ratios using a small number of rounds. We first show that $O(\frac {\log n} {\delta^2})$ rounds suffice to get a $(1+\delta)$ approximation. 
In Subsection \ref{subsec-k-round-matching} we present an algorithm that provides an approximation ratio of $O(n^{\frac 1 {k+1}})$ 
in $k$ rounds. This shows that even a constant number of rounds suffices to get much better approximation ratios than what can be achieved by simultaneous algorithms.

\subsection{A $(1+\delta)$-Approximation for Bipartite Matching in $O(\frac {\log n} {\delta^2})$ Rounds}

The algorithm is based on an auction where each player competes at every point on one item that he demands the most at the current prices. Therefore, it will be easier for us to imagine the players as having valuations. Specifically, each player $i$ is a unit demand bidder with $v_i(j)=1$ if $j\in S_i$ and $v_i(j)=0$ otherwise. 

\subsubsection*{The Algorithm}

\begin{enumerate}
\item For every item $j$ let $p_j=0$. 
\item Let $N_1$ be the set of all players. 
\item In every round $r=1,\ldots,\frac {2\log n} {\delta^2}$:
\begin{enumerate}
\item For each player $i\in N_r$, let the \emph{demand} of $i$ be $D_i=\arg\min_{j:p_j<1, j\in S_i}{p_j}$. This is the subset of $S_i$ for which the price of each item is minimal and smaller than $1$. 
\item Each player $i \in N_r$ selects uniformly at random an item $j_i\in D_i$ and reports its index.
\item\label{step-order} Go over the players in $N_r$ in a fixed arbitrary order. If item $j_i$ was not yet allocated in this round, player $i$ receives it and the price $p_{j_i}$ is increased by $\delta$. In this case we say that player $i$ is committed to item $j_i$. A player $i'$ that was committed to $j_i$ in the previous round (if such exists) now becomes uncommitted. 
\item Let $N_{r+1}$ be the set of uncommitted players at the end of round $r$. 
\end{enumerate}
\end{enumerate}

Our algorithm is very similar to the classical auction algorithms except for two seemingly small changes. However, quite surprisingly, these changes allow us to substantially reduce the communication cost. The first change is to ask all the players to report an item of their demand set simultaneously (instead of sequentially). This change alone is not enough as in the worst case many players might report the same item and hence the number of rounds might still be $\Omega(\frac{n}{\delta})$. Hence we ask each player to report a \emph{random} item of his demand instead. 
 

\begin{theorem} \label{thm-auction-alg}
After $O(\frac {\log n} {\delta^2})$ rounds the algorithm above provides an approximation ratio of $(1+\delta)$. 
\end{theorem}
\begin{proof}
Fix some optimal solution $(o_1,...,o_n)$ (every player receives at most one item). Let $N'$, $|N'|=n'$, be the set of players that receive an item in the optimal solution. 
\begin{definition}
A player is called \emph{satisfied} if he is either allocated an item or $D_i=\emptyset$.
\end{definition}

Let $END$ be the random variable that denotes the number of rounds until the first time that $(1-\delta)n'$ players in $N'$ are satisfied. We will prove the following lemma:

\begin{lemma}
$E[END]\leq \frac {4\log n} {\delta^2}$.
\end{lemma}
\begin{proof}
The heart of the proof is the definition of two budgets: one for demand halving actions for players in $N'$ and one for price increments. We show that in expectation after at most $\frac {4\log n} {\delta^2}$ rounds at least one of these budgets is exhausted and hence the number of unsatisfied players in $N'$ is at most $\delta n'$. 
 
Consider the demand set $D_i$ of some player $i$ at some round $r$. Observe that all items in $D_i$ has the same price $p_{D_i}$. Let $D^{p}_i=S_i\cap \{ j|p_j=p \}$.
We will use the following claim:
\begin{claim}
Consider some round $r$ and suppose that at least $t$ players are unsatisfied in the beginning of that round. Then, either the expected increase in $\Sigma_j p_j$ is at least $t \cdot \frac {\delta} {4}$ or for at least half of the unsatisfied players it holds that $D^{p_{D_i}}_i$ has shrunk by at least a factor of $2$.
\end{claim}
\begin{proof}
Consider some player $i$ that is not satisfied. When it is $i$'s turn to be considered in Step \ref{step-order} either at least half the items in $D_i$ were taken by previous players in the order or not. If at least half the items in $D_i$ were taken by previous players then $D^{p_{D_i}}_i$ has shrunk by a factor of at least $2$. Otherwise, since player $i$ selects $j_i$ at random from $D_i$, with probability of at least $\frac 1 2$ we have that $j_i$ was not taken by any previous player. In this case $i$ the price of $p_{j_i}$ is increased by $\delta$ so the expected increase of some item due to $i$ is $\frac \delta 2$. This implies that either for at least $t/2$ players $D^{p_{D_i}}_i$ has shrunk by at least a factor of $2$ or at least $t/2$ players caused an expected increase of $\delta/2$ in the price of some item. The later implies by linearity of expectation an expected total increase of $t \cdot \frac {\delta} {4}$.
\end{proof}

Now, notice that the price of each item can be increased at most $\frac 1 \delta$ times (the price increases in increments of $\delta$ and no player demands an item which has a value of $1$). Since an item that was allocated stays allocated and at most $n'$ items can be allocated, we have that the maximal number of increments that the algorithm can make is $\frac {n'} \delta$.

In addition, there are $n$ items and the price of an item can only increase, each $D^p_i$ can be shrunk by a factor of $2$ at most $\log n$ time. As previously argued, $p$ can get only $\frac 1 \delta$ different values. The total number of shrinkage steps with respect to players in $N'$ is therefore $\frac {n'\log n} \delta$.

To complete the proof recall that at every round prior to $END$ at least $\delta n'$ players in $N'$ are unsatisfied. By the claim in each round either the expected number of increments is $\frac {\delta n'} 4$ or the expected number of shrinkage steps with respect to players in $N'$ is $\frac {\delta n'}{2}$. In any case, after $\frac {4\log n}{\delta^2}$ rounds we expect that there are no more increments or shrinkage steps with respect to players in $N'$ to make.\end{proof}

\begin{lemma}
If at least $(1-\delta)n'$ of the players in $N'$ are satisfied then the approximation of the algorithm is $1-2\delta$.
\end{lemma}
\begin{proof}
We use a variant of the first welfare theorem to prove the lemma. Consider a player $i$ that has received an item $j_i$. The player receives an item that maximizes his demand (up to $\delta$), and thus the profit from $j_i$ is at least the profit from the item $o_i$ he got in the optimal solution (up to $\delta$, we allow $o_i=\emptyset$). We therefore have: $v_i(j_i)-p_{j_i}\geq v_i(o_i)-p_{o_i}-\delta$. For each satisfied player $i$ that did not receive any items we have that $0\geq v_i(o_i)-p_{o_i}$. Denote by $N_s$ the set of satisfied players. Summing over all satisfied players we get:

\begin{align*}
\sum_{i \in N_s}(v_i(j_i)-p_{j_i}) &\geq \sum_{\text{$i\in N_s$ is allocated}}(v_i(o_i)-p_{o_i} -\delta) ~~+\sum_{\text{$i\in N_s$ is unallocated}}(v_i(o_i)-p_{o_i})\\
&= \sum_{i \in N_s} (v_i(o_i)-p_{o_i}) - \delta n' \\
ALG - \sum_{j\in N}p_j &\geq OPT - \delta n' - \sum_{i\in N_s}p_{o_i} - n'\delta\\
ALG &\geq OPT - 2n'\delta = (1-2\delta)n'
\end{align*}
where in the third transition we used the facts that items that are unallocated by the algorithm have a price of $0$ and that $N' \cap N_s \geq (1-\delta)N'$.
\end{proof}
\end{proof}

It is worth noting a different version of the auction algorithm which was discussed in \cite{DGS86}. In this version at every round each player reports its entire demand set (simultaneously with the other players), then a minimal set of over demanded items is computed and only their prices are increased. While the number of rounds for this algorithm might be small the communication cost of each round can be linear in $n$.

\subsection{A $k$-Round Algorithm for Matching}\label{subsec-k-round-matching}

Fix some optimal solution $(o_1,...,o_n)$ (every player receives at most one item). Let $N'$ be the set of players that receive a nonempty bundle in the optimal solution. The following algorithm achieves an approximation ratio of $O(n^{\frac 1 {k+1}})$ in $k$ rounds.

\subsubsection*{The Algorithm}
\begin{enumerate}
\item Let $N_1=V_1$ and $U_1=V_2$.
\item In every round $r=1,\ldots,k$:
\begin{enumerate}
\item Each player $i$ selects uniformly at random an item $j_i\in U_i$ that he demands.
\item\label{step-order-k} Go over the players in $N_r$ in a fixed arbitrary order. Player $i$ receives $j_i$ if this item was not allocated yet.
\item Let $N_{r+1} \subseteq N_r$ be the set of players that were not allocated items at round $r$ or before.
\item Let $U_{r+1} \subseteq U_r$ be the set of items that were not allocated at round $r$ or before.
\end{enumerate}
\end{enumerate}

\begin{theorem}
For every $k\leq \log n$, the approximation ratio of the algorithm is $O(n^{\frac 1 {k+1}})$. In particular, when $k=1$ the approximation ratio is $O(\sqrt n)$ and when $k=O(\log n)$ the approximation ratio is $O(1)$.
\end{theorem}
\begin{proof}
Consider a run of the algorithm. Let $D_{r,i}\subseteq U_r$ be the set of items that player $i$ demands and are still available immediately before round $r$ starts. Let $X_{r,i}$ be the set of items that where allocated to other players before $i$'s turn in step (\ref{step-order-k}). Player $i$ is said to be \emph{easy to satisfy} if in some round $r$ we have that $D_{r,i}- X_{r,i}\geq \frac {D_{r,i}} {n^{\frac 1 {k+1}}}$. Let $S$ be the event that at least half of the players in $N'$ are easy to satisfy. We will show that $E[ALG|S]=O(\frac {OPT} {n^{\frac 1 {k+1}}})$ and that $E[ALG|\bar S]=O(\frac {OPT} {n^{\frac 1 {k+1}}})$, where $ALG$ is the random variable that denotes the value of the solution that the algorithm outputs. Together this implies that $E[ALG]=O(\frac {OPT} {n^{\frac 1 {k+1}}})$. Each one of the next two lemmas handles one of those cases.

\begin{lemma}
$E[ALG|S]=O(\frac {OPT} {n^{\frac 1 {k+1}}})$.
\end{lemma}
\begin{proof}
Let $C_{i,r}$ be the random variable that denotes the probability that player $i$ is allocated an item at round $r$. Observe that if player $i$ is easy to satisfy then for some round $r$ we have that $E[C_{i,r}]\geq \frac {1} {n^{\frac 1 {k+1}}}$. Let $P$ denote the set of easy to satisfy players. The expected number of easy to satisfy players that are allocated an item is at least $E[\Sigma_{i\in P} C_{i,r}]=\Sigma_{i\in P} E[C_{i,r}]\geq |P|\cdot \frac {1} {n^{\frac 1 {k+1}}} \geq \frac {n'} 2 \cdot \frac {1} {n^{\frac 1 {k+1}}}$. The required approximation ratio follows since the value of the optimal solution is $n'$.
\end{proof}

\begin{lemma}
$E[ALG|\bar S]=O(\frac {OPT} {n^{\frac 1 {k+1}}})$.
\end{lemma}
\begin{proof}
Consider a player $i$ that is not easy to satisfy. Observe the for every such player $i$ and round $r$, if $i \in N_r \cap N_{r+1}$ then $|D_{r,i}| \geq n^{\frac{1}{k+1}}| D_{r+1,i}|$. This is true since if such player $i$ was not allocated any items at round $r$ then the set of available items that he demands shrinks. Therefore, for every player $i$ that was not allocated anything at round $k$, we have that $|D_{k+1,i}| \leq n^{\frac{1}{k+1}}$. If there exists a player $i$ such that $|D_{k+1,i}| >0$ then it has to be the case that initially $|D_{1,i}|\geq n^{\frac k {k+1}}$. This implies that at least $n^{\frac k {k+1}}-n^{\frac{1}{k+1}}$ of the items were allocated and hence the approximation ratio is $O(n^{\frac 1 {k+1}})$ (there are $m$ items, so the value of the optimal solution is at most $m$). In any other case for every such player $i$ that was not allocated any bundle we have that $D_{k+1,i}=\emptyset$. In particular, the item $o_i$ that he receives in the optimal solution was allocated. Since there are at least $\frac {n'} 2$ such players, this implies that at least $\frac {n'} 2$ items were allocated and proves the claimed approximation bound in this case as well. 
\end{proof}

\end{proof}

\section{A Lower Bound for Subadditive Combinatorial Auctions}\label{sec-xos-lb}
We now move to discuss combinatorial auctions with subadditive bidders. In particular, in this section we prove our most technical result:
\begin{theorem}\label{thm-xos-randomized-lb}
No randomized simultaneous protocol for combinatorial auctions with subadditive bidders where each bidder sends sub-exponentially many bits can approximate the social welfare to within a factor of $m^{\frac 1 4 -\epsilon}$, for every constant $\epsilon>0$.
\end{theorem}

We will actually prove the lower bound using only XOS valuations, a subclass of subadditive valuations. We present a distribution over the inputs and show that any deterministic algorithm in which each bidder sends sub-exponentially many 
bits cannot approximate the expected social welfare to within a factor of $m^{1/4-\epsilon}$ for this 
specific distribution (where expectation is taken over the distribution). This implies, by Yao's principle, that there is no randomized algorithm that achieves an approximation ratio of $m^{1/4-\epsilon}$ for every constant $\epsilon>0$. 


Our hard distribution which we denote by $D$ is the following:
\begin{itemize}
\item $n=k^3$ players, $m=k^3+k^4$ items. 
\item Each player $i$ gets a family $F_i$ of size  $t=e^{2k^\epsilon}$ of sets of $k$ items. The valuation of player $i$ is defined to be: $v_i(S) = \max_{T \in F_i} |T \cap S|$. Observe that this is an XOS valuation where each set $S \in F_i$ defines a clause in which all items in $S$ have value $1$ and the rest of the items have value $0$.
\item The families $F_i$ are chosen, in a correlated way, as follows: first, a center $C$ of size $k^3$ is chosen at random; then for each player $i$ a petal $P_i$ of size $k^2$ is  chosen at random from the complement of $C$. Now, for  each player $i$, the family $F_i$  is chosen as follows: one set $T_i$ of size $k$ is chosen at random from $P_i$ and $t-1$ sets of size $k$ are chosen at random from $C \cup P_i$.
\item The players do not know
$C$ nor do they know which of the sets was chosen from $P_i$.  We may assume without loss of generality that each player $i$ knows
the set $C \cup P_i$.
\end{itemize}

Each player sends, deterministically, simultaneously with the others, at most $l$ bits of communication just based on his input $F_i$.  A referee that sees all the messages chooses an allocation $A_1, \ldots, A_n$ of the $m$ items to the $n$ players (only based on the messages), with the $A_i$'s being disjoint sets. We assume without loss of generality that all items are allocated. 

In order to prove that no deterministic algorithm can obtain a good approximation for instances drawn from $D$, we show that to get a good approximation we must identify $T_i$ for almost all of the players. This would have been easy had each player could have distinguished between the items in $C$ and $P_i$, but this information is missing. We show that for the central planner to successfully identify even a single $T_i$ player $i$ has to send exponentially many bits. Formally, this is done by reducing the two-player ``set seeking'' problem that we define below to the multi-player combinatorial auction problem. The main technical challenge is to prove the hardness of the set seeking problem.


The next couple of lemmas together gives us Theorem \ref{thm-xos-randomized-lb}. The proof of the first lemma is easy, but the second lemma is the heart of the lower bound.


\begin{lemma}\label{lemma-xos-lb-opt}
With very high probability, over this distribution $(D)$, there exists an allocation with social welfare
$\sum_i v_i(A_i) = \Theta(k^4)$.
\end{lemma}
\begin{proof}
Consider allocating each player $i$ the set $T_i$ (an item $j$ that is in multiple $T_i$'s will be allocated to some player $i$ such that $j \in T_i$). The social welfare of this allocation is 
$|\cup_i T_i|$. We show that with high probability the social welfare is $\Theta(k^4)$. This follows since each $T_i$ is of size $k$ and is practically selected uniformly at random from a set of size $k^4$. Thus, the probability that item $j$ is in $T_i$ is $1/k^3$. Since there are $k^3$ players, an item is in $\cup_i T_i$ with constant probability. Next, since items are chosen independently we can use Chernoff bounds to get that with high probability a constant fraction of all items is in $\cup_i T_i$.
\end{proof}

\begin{lemma}\label{lemma-xos-det-lb}
Every deterministic protocol for the combinatorial auction problem with $l<t^\epsilon$ produces an allocation with $\sum_i v_i(A_i) = k^{3+O(\epsilon)}$, in expectation (over the distribution).
\end{lemma}

To prove Lemma \ref{lemma-xos-det-lb} we first define a two-player ``set seeking'' problem and show its hardness (Subsection \ref{subsec-set-seeking}). Next, we reduce the two-player set seeking problem to our multi-player combinatorial auction problem (Subsection \ref{subsec-xos-reduction}). 

\subsection {The Two Player ``Set Seeking'' Problem} \label{subsec-set-seeking}
The ``Set Seeking'' problem includes two players and $x=k^2+k^3$ items. One of the players plays the role of the keeper and gets as an input a family of $t=e^{2k^\epsilon}$ sets $F$ where all the sets are of size $k$. The other player plays the role of the seeker and gets a set $P$ of size $k^2$. In this problem, first the keeper sends a message of at most $l$ bits (advice). Next, based on this message the seeker outputs some set $A\subseteq P, |A|\geq k$. The goal is to maximize $\max_{T \in F} \dfrac{|A \cap T|}{|A|}$.  

We will analyze the performance of deterministic algorithms on a specific distribution $(D_2)$ for this problem which we now define. This distribution is based on choosing two sets, $F$ and $P$ which are chosen in correlation as follows:
\begin{enumerate}
\item A set $P$ of size $k^2$ is chosen uniformly at random from all items.
\item The set $F$ is constructed by choosing uniformly at random a special set $T_P$ of size $k$ from $P$ and additional $t-1$ sets of size $k$ from all items.
\end{enumerate}

\begin{lemma}\label{lemma-set-seeking}
If $l\leq t^{\epsilon}$ then there is no $k^{1-\epsilon}$-approximation for the set seeking problem.
\end{lemma}
\begin{proof}
Fix a message $m$ and let $A_m(P)$ denote the set $A$ that the seeker returns when his input is $P$ and the keeper sends a message $m$.
\begin{definition}
Fix a message $m$ and a set $P$, $|P|=k^2$. A set $S$ is \emph{$(m,P)$-compatible} if $|A_m(P) \cap S| \geq \dfrac{|A_m(P)|}{k^{1-\epsilon}}$.
\end{definition}

\begin{claim}\label{claim-m-compatible}
Fix a message $m$ and a set $P$, $|P|=k^2$. The probability that a set $S$ which is chosen uniformly at random from $P$ is $(m,P)$-compatible is at most $2e^{-k^\epsilon}$.
\end{claim}
\begin{proof}
We fix a set $A_m(P)=A\subseteq P$ and compute the probability that the intersection of this set with a set $S$ of size $k$ chosen uniformly at random out of $k^2$ items will be greater than $\dfrac{|A|}{k^{1-\epsilon}}$. The probability of each element in $S$ to be in $A \cap S$ is exactly $\dfrac{|A|}{k^2}$. Thus we expect that $|A \cap S| \approx \dfrac{|A|}{k}$. We now use Chernoff bounds to make this precise. 

Consider constructing the following random set $T$: $k$ items are selected so that each item is chosen uniformly at random amount the $k^2$ items of $P$. $T$ is similar to the way $S$ is constructed, except that it possibly contains less than $k$ items as there is some positive chance that some item in $P$ will be selected twice. We conservatively assume that every item that was selected twice is in $A$. Thus, if we bound the probability that $|A\cap T|$ is too large, we also bound the probability that $|A\cap S|$ is too large.

We first bound the probability that more than $k^{\epsilon}$ items are selected at least twice to $T$. Since there are at most $k$ items in $T$, in the $k$'th item that we have the select, the probability that we will choose an already-selected item is at most $\frac {k-1} {k^2}\leq \frac 1 k$. By the Chernoff bounds, the probability that more than $k^{\epsilon}$ items are selected twice to $T$ is at most: $\frac{1}{e^{ {k^\epsilon} }}$.

To compute the expected intersection with $A$, assume $k$ independent variables, each variable $y_i$ is true with probability $\dfrac{|A|}{k^2}$ and false with probability $1-\dfrac{|A|}{k^2}$. Observe that $Y=\sum_i y_i$ is distributed exactly as $|T\cap A|$. The expectation of $Y$ is $\dfrac{|A|}{k}$. Now by Chernoff bounds we get that: $\Pr[Y> k^\epsilon \cdot \dfrac{|A|}{k} ] \leq \left (\dfrac{1}{e^{k^\epsilon}}\right )^{\dfrac{|A|}{k}} \leq  \dfrac{1}{e^{k^\epsilon}}$. The last transition holds by the assumption that $|A|\geq k$.

We have that with probability of at most $\dfrac{2}{e^{k^\epsilon}}$ we have that $|A\cap T|\leq 2e^{-k^\epsilon}$. By our discussion above, with at most the same probability we have that $|A\cap S|\leq 2e^{k^\epsilon}$.
\end{proof}

\begin{definition} \label{def-f-good}
Let $F$ be a family of sets of size $k$, $|F|=t=e^{2k^\epsilon}$. We say that a message $m$ is \emph{$F$-good} if $\Pr[|A_m(P)\cap T|\geq \dfrac {|A_m(P)|} {k^{1-\epsilon}}]\geq e^{-\frac 1 2 k^\epsilon}$, where $T$ is chosen uniformly at random from $F$ and $P$, $|P|=k^2$, contains $T$ and $k^2-k$ items chosen uniformly at random from the rest of the items.
\end{definition}

\begin{claim}
For every message $m$, the probability that $m$ is $F$-good is at most \[p'=e^{{(-\frac 1 2k^\epsilon)}^{(e^{k^\epsilon)}}}\]where the sets in $F$ are chosen uniformly at random.
\end{claim}
\begin{proof}
Fix a message $m$. Consider $F$ where the sets in $F$ are chosen uniformly at random. We first compute the probability that for a single $T \in F$ we have that $|A_m(P)\cap T|\geq \dfrac {A_m(P)} {k^{1-\epsilon}}$. Observe that every set $T \in F$ in this setting can be thought of as chosen uniformly at random from a fixed set $P$. Thus that probability is the same as the probability that $T$ is $(m,P)$-compatible, which is $e^{-k^\epsilon}$ by Claim \ref{claim-m-compatible}.

Now we would like to compute the probability that $m$ is $F$-good, that is the probability that there exist at least $e^{-\frac 1 2 k^\epsilon}\cdot |F| = e^{\frac 3 2 k^\epsilon}$ sets $T\in F$ such that  $|A_m(P)\cap T|\geq \dfrac {|A_m(P)|} {k^{1-\epsilon}}$. The expected number of such sets is $e^{-k^\epsilon} \cdot |F|=e^{k^\epsilon}$. By the Chernoff bounds ($\mu=e^{k^\epsilon}$, $\delta=e^{\frac 1 2 k^\epsilon}$) this probability is at most $p'$.
\end{proof}

We can now finish the proof. Choose $F$ at random. For every $m$ the probability that $m$ is $F$-good is at most $p'$. The message length is $l$, and the total number of messages is therefore at most $2^l$. Thus, by the union bound, the probability that there exists some message $m$ which is $F$-good is at most $2^l\cdot p'\leq 2^{t^\epsilon}  \cdot p' = 2^{e^{k^{\epsilon^\epsilon}}}  \cdot p'< (2e^{(-\frac 1 2k^\epsilon)})^{{(e^{k^\epsilon)}}} <e^{-e^{k^\epsilon}}$, where the transition before the last uses the fact that $\epsilon<1$. Hence, with probability at least $1-e^{-e^{k^\epsilon}}$ every message $m$ is not $F$-good for the randomly chosen $F$. This in turn implies that for a family $F$ and for every message $m$ the probability for $P$ and $T$ chosen as in Definition \ref{def-f-good} that $T$ is $(m,P)$-compatible is at most $e^{-\frac 1 2 k^\epsilon}$. Thus, with probability $1-e^{-e^{k^\epsilon}}-e^{-\frac 1 2 k^\epsilon}$ we have that $|A_m(P)\cap T|\leq \dfrac {|A_m(P)|} {k^{1-\epsilon}}$. This implies that with probability $1-e^{-e^{k^\epsilon}}-e^{-\frac 1 2 k^\epsilon}$ the approximation ratio is at most $k^{1-\epsilon}$. Even if with probability $e^{-e^{k^\epsilon}}+e^{-\frac 1 2 k^\epsilon}$ the approximation ratio is $1$, the expected approximation ratio is at most $\frac{1}{2}k^{1-\epsilon}$. 
\end{proof}

\subsection{The Reduction (Set Seeking $\rightarrow$ Combinatorial Auctions with XOS bidders)} \label{subsec-xos-reduction}

Given the hardness of the set seeking problem, we will be able to derive our result for combinatorial auctions using the following reduction:

\begin{lemma}\label{lemma-xos-reduction}
Any protocol for combinatorial auctions on distribution $D$ that achieves approximation ratio better than $m^{1/4-\epsilon} $ where the message length of each player is $l$ can be converted into a protocol for the two-player set seeking problem on distribution $D_2$ achieving an approximation ratio of $k^{1-\epsilon}$ with the same message length $l$.
\end{lemma}
\begin{proof}
In the proof we fix an algorithm for the multi-player combinatorial auction problem and analyze its properties.
\begin{definition}
Fix an algorithm for the XOS problem and consider the distribution $D$. We say that player $i$ is \emph{good} if $E[|A_i\cap T_i|]\geq \max \{ \dfrac{|A_i|}{k^{1-\epsilon}}, k^{\epsilon} \}$.
\end{definition}
To prove the lemma we first show that if none of the players are good then the algorithms approximation ratio is bounded by $m^{1/4-\epsilon} $. Else, there exists at least a single player which is good. In this case we show how the algorithm for combinatorial auctions can be used to get a good approximation ratio for the set seeking problem.

\begin{claim}
If none of the players is good then the expected approximation ratio is at most $m^{\frac 1 4 - \epsilon}$.
\end{claim}
\begin{proof}
To give an upper bound on the expected social welfare we assume that the $k^3$ items in the center are always allocated to players that demand them. We now compute an upper bound on the contribution of the remaining $k^4$ items to the expected social welfare.
Observe that since none of the players is good, each player contributes at most $\dfrac{|A_i|}{k^{1-\epsilon}}+ k^{\epsilon}$ to the expected social welfare (of the $k^4$ items). Hence the expected social welfare achieved by the algorithm is at most $\sum_i(\dfrac{|A_i|}{k^{1-\epsilon}}+k^\epsilon) +|C|  = \dfrac{k^4}{k^{1-\epsilon}} + k^3\cdot k^\epsilon +k^3 \leq 3k^{3+\epsilon}$. This implies that the approximation ratio of the algorithm is $k^{1-\epsilon} \leq m^{\frac 1 4 - \epsilon}$.
\end{proof}

\begin{claim}
If there exists a good player then there exists an algorithm for the two-player set seeking problem that guarantees an approximation ratio of $k^{1-\epsilon}$ with the same message length.
\end{claim}
\begin{proof}
Let player $i$ be the good player. Recall that $D$ is the distribution for the multi-player combinatorial auction problem and that $D_2$ is the distribution defined for the set seeking problem. We denote by $E_{D}[\cdot]$ and $E_{D_2}[\cdot]$ expectations taken over the distributions $D$ and $D_2$ respectively. We show that there exists an algorithm for the set seeking problem achieving expected approximation ratio of $k^{1-\epsilon}$ on $D_2$. 

Let the keeper take the role of player $i$ in the multi-player algorithm and the seeker play the roles of the rest of the $n-1$ players. More precisely, the keeper first sends player $i$'s message to the seeker. This is possible as the input of the keeper is identical to the input of the players in the multi-player problem. Next, the seeker simulates the messages of the remaining players and run the algorithm internally. This simulation is possible by the assumption that in the multi-player algorithm all messages are sent simultaneously. The number of items in the combinatorial auction will be $k^3+k^4$, where the $x$ items of the set seeking problem will correspond to some set $X$ of size $x$ of items in the combinatorial auction. We first show that given that the information of the seeker and keeper is drawn in a correlated way from $D_2$, they have enough information to simulate the correlated distribution $D$. 

The input of the keeper is defined in a straightforward way, where each set in $F$ defines a clause in the XOS valuation. All items in these clauses are subsets of $X$. The seeker constructs the valuations of the other $n-1$ players as follows: the items that are in $X \setminus P$ form the center. Next the seeker chooses uniformly at random for each player $j$ a petal $P_j$ of the $k^4$ items not in the center, a set $T_j \subseteq P_j$ of size $k$ and additional $t-1$ sets of size $k$ from $C\cup P_j$. Observe that the distribution of valuations constructed this way is identical to $D$. The inherent reason for this is that for player $i$ the distribution of $P_i$ and $T_i$ is identical to the distribution of $P$, $T$ and $F$ in $D_2$ as in both cases $P$ and the $t-1$ sets in $F$ (or $P_i$ and $F_i$) are chosen uniformly at random from a set of size $k^3+k^2$ and $T$ (or $T_i$) is chosen uniformly at random from $P$ (or $P_i$). In other words, the distribution $D_2$ on $F,P$ and $T$ is identical to distribution $D$ projected on $F_i,P_i$ and $T_i$.

We now observe that since $i$ is a good player we have that $E_D[|A_i\cap T_i|]\geq \max \{ \dfrac{|A_i|}{k^{1-\epsilon}}, k^{\epsilon} \}$. We show that this implies an algorithm achieving expected approximation ratio of $k^{1-\epsilon}$ for the two-player set seeking problem. The algorithm works as follows: we first perform the reduction above, and therefore the distribution we are analyzing is $D$. Now, if player $i$ was assigned a bundle $A_i$ of size at least $k$, the algorithm returns $A=A_i$. Else, the algorithm returns a bundle $A$ that contains $A_i$ and additional $k-|A_i|$ arbitrary items. Thus, we have that $E_D[|A\cap T_i|]\geq \dfrac{|A|}{k^{1-\epsilon}}$, as we made sure that $|A|\geq k$ implying that $\dfrac{|A|}{k^{1-\epsilon}}\geq k^\epsilon$.

We claim that the expected approximation ratio of the algorithm on $D_2$ is $k^{1-\epsilon}$. As the distribution $D_2$ on $F,P$ and $T$ is identical to distribution $D$ projected on $F_i,P_i$ and $T_i$, we have that $E_{D_2}[|A\cap T|]\geq \dfrac{|A|}{k^{1-\epsilon}}$. Thus, $E_{D_2}[\dfrac{|A\cap T|}{|A|}]\geq \dfrac{1}{k^{1-\epsilon}}$. This in turn implies that $\max_{S \in F} \dfrac{|A \cap S|}{|A|} \geq \dfrac{1}{k^{1-\epsilon}}$ and since the optimal solution has a value of $1$ the expected approximation ratio is $k^{1-\epsilon}$.
\end{proof}

From the last two claims we get that either the algorithm for the combinatorial auction problem does not guarantee a good approximation ratio, or that we have constructed an efficient protocol for the set seeking problem.
\end{proof}

\section{Algorithms for Subadditive Combinatorial Auctions} \label{sec-xos-alg}

We design algorithms for a restricted special case of ``$t$-restricted'' instances (see definition below). We will show however that the existence of a simultaneous algorithm for $t$-restricted instances implies a simultaneous approximation algorithm for subadditive bidders with almost the same approximation ratio.

\begin{definition}
Consider an XOS valuation $v(S)=\max_r a_r(S)$, where each $a_r$ is an additive valuation. $v$ is called \emph{binary} if for every $a_r$ and item $j$ we have that $a_r(\{j\})\in \{0,\mu\}$, for some $\mu$.
\end{definition}

\begin{definition}
An instance of combinatorial auctions with binary XOS valuations (all with the same $\mu$, for simplicity and without loss of generality $\mu=1$) is called \emph{$t$-restricted} if there exists an allocation $(A_1,\ldots , A_n)$ such that all the following conditions hold:
\begin{enumerate}
\item For every $i$, $v_i(A_i)=|A_i|$.
\item For every $i$, either $|A_i| = t$ or $|A_i|=0$.
\item $t$ is a power of $2$.
\item $\Sigma_iv_i(A_i)\geq \dfrac {OPT} {2\log m}$.
\end{enumerate}
\end{definition}

\begin{proposition}\label{prop-restricted-instances}
If there exists a simultaneous algorithm for $t$-restricted instances that provides an approximation ratio of $\alpha$ where each bidder sends a message of length $l$, then there exists a simultaneous algorithm for subadditive bidders that provides an approximation ratio of $O(\alpha\cdot \log^3 m)$ where each bidder sends a message of length $O(l \cdot \log^3 m)$.
\end{proposition}
\begin{proof}
The proposition follows from the following three lemmas.

\begin{lemma}
If there exists a simultaneous algorithm for XOS bidders that provides an approximation ratio of $\alpha$ where each bidder sends a message of length $l$, then there exists a simultaneous algorithm for subadditive bidders that provides an approximation ratio of $O(\alpha\cdot \log m)$ where each bidder sends a message of length $l$.
\end{lemma}
\begin{proof}
For every subadditive valuation there is an XOS valuation that is an $O(\log m)$ approximation of it \cite{D07}\footnote{A valuation $v$ $\alpha$-approximates  a valuation $u$ if for every $S$ we have that $u(S)\geq v(S)\geq \frac{u(S)} \alpha$.}; thus, if each player computes this XOS valuation and proceeds as in the algorithm for XOS valuations we get an algorithm for subadditive valuations, losing only an $O(\log m)$ factor in the approximation ratio.
\end{proof}

\begin{lemma}
If there exists a simultaneous algorithm for binary XOS bidders (all with the same $\mu$) that provides an approximation ratio of $\alpha$ where each bidder sends a message of length $l$, then there exists a simultaneous algorithm for XOS bidders that provides an approximation ratio of $O(\alpha\cdot \log m)$ where each bidder sends a message of length $O(l \cdot \log m)$.
\end{lemma}
\begin{proof}
We will move from general XOS valuations to binary XOS valuations using the following notion of projections:
\begin{definition}
A \emph{$\mu$-projection} of an additive valuation $a'$ is the following additive valuation $a$:
\begin{eqnarray*}
a(\{j\})=\left\{
  \begin{array}{ll}
    \mu, & \hbox{if $2\mu>a'(j)\geq \mu$;} \\
    0, & \hbox{otherwise.}
  \end{array}
\right.
\end{eqnarray*}
\end{definition}
A \emph{$\mu$-projection} of an XOS valuation $v$ is the XOS valuation $v^\mu$ that consists exactly of all $\mu$-projections of the additive valuations (clauses) that define $v$.

Let $v_{max}$ be $\max_i v_i(M)$ rounded down to the nearest power of $2$. Let $\mathcal M=\{ \frac {v_{max}} {2m}, \frac {v_{max}} {m}, \ldots, \frac {v_{max}} 2, v_{max}\}$. The next claim shows that there exists some $\mu\in \mathcal M$ such that the value of the optimal solution with respect to the $\mu$-projections $v_i^\mu$ is only a logarithmic factor away from the value of the optimal solution with respect to the $v_i$'s. Given the claim below and an algorithm $A$ for binary XOS valuations we can construct the following algorithm for general XOS valuations: each player $i$ computes, for every $\mu\in \{ \frac {MAX_i} {2m}, \frac {MAX_i} {m},  \ldots, \frac {MAX_i} 2, MAX_i\}$, where $MAX_i$ equals $v_i(M)$ rounded down to the nearest power of $2$, his $\mu$-projection and sends both $\mu$ and the message of length $l$ he would have sent in the algorithm $A$ if his valuation was $v_i^\mu$. The new algorithm computes now up to $|\mathcal M|$ different allocations by running $A$ once for each $\mu\in \mathcal M$, and outputs the allocation with the best value\footnote{Observe that if player $i$ does not report a valuation for some value of $\mu \in \mathcal M$ then the algorithm may assume its valuation for this $\mu$ is $0$ for every bundle.}. By the claim below, the approximation ratio is $O(\alpha \cdot \log m)$. The length of the message that each bidder sends is $O(l\cdot \log m)$.

\begin{claim}
Let $OPT$ be the value of the optimal solution with respect to the $v_i$'s. For every $\mu\in \mathcal M$, let $OPT_\mu$ be the value of the optimal solution with respect to the $\mu$-projections $v_i^\mu$. There exists some $\mu \in \mathcal M$ such that $OPT_\mu \geq \dfrac {OPT} {8\log m}$.
\end{claim}
\begin{proof}
Fix some optimal solution $(O_1,\ldots, O_n)$. For each player $i$ let $a_i$ be the maximizing clause for the bundle $O_i$ in $v_i$. Now, for every player $i$ and every item $j\in O_i$, put item $j$ into bin $x$, where $x$ is a power of $2$, if and only if $2x>a_i(\{j\})\geq x$. Let $M_x$ be the set of items in bin $x$. We claim that there exists bin $\mu$, $\mu \geq \frac {v_{max}} {2m}$, for which it holds that $\Sigma_i \Sigma_{j\in  O_i \cap M_\mu }a_i(\{j\})\geq \frac {\Sigma_i v_i(O_i)} {4\log m}$. To see this, first let $L$ be the set of ``small'' items that are in any of the bins $x$, $x\leq \frac {v_{max}} {4m}$. It holds that $\Sigma_i\Sigma_{j\in O_i \cap L}a_i(\{j\})\leq \Sigma_i\Sigma_{j\in O_i \cap L} \frac{v_{max}}{2m}\leq \frac {v_{max}} 2 \leq \frac {\Sigma_iv_i(O_i)} 2$. Thus, we have that $\Sigma_{i} v_i(O_i \cap \cup_{x\in \mathcal M} M_x) \geq \frac {\Sigma_iv_i(O_i)} 2$. Now observe that the number of bins $x$, $x\geq  \frac {v_{max}} {2m}$ is bounded from above by $2\log m$. Therefore, there exists a bin $\mu$ such that $\Sigma_{i} v_i(O_i \cap M_\mu) \geq \frac {\Sigma_iv_i(O_i)} {4 \log m}$. The proof is completed by observing that $\Sigma_{i} v^{\mu}_i(O_i \cap M_\mu) \geq \frac{1}{2} \Sigma_{i} v_i(O_i \cap M_\mu)$ as the $\mu$-projection cuts the value of each of the item by at most half.

\end{proof} 

\end{proof}

\begin{lemma}
If there exists a simultaneous algorithm for $t$-restricted instances that provides an approximation ratio of $\alpha$ where each bidder sends a message of length $l$, then there exists a simultaneous algorithm for binary XOS bidders (with the same $\mu$) that provides an approximation ratio of $O(\alpha\cdot \log m)$ where each bidder sends a message of length $O(l \cdot \log m)$.
\end{lemma}
\begin{proof}
We will show that for every instance of combinatorial auctions with binary XOS bidders there exists some $t$ for which  this instance is $t$-restricted. Given algorithms $A_t$ for $t$-restricted instances we can construct an algorithm for binary XOS valuations as follows: for each $t\in \{1,2,4, \ldots, m\}$ each bidder sends the same message as in $A_t$. Now compute $\log m+1$ allocations, one for each value of $t$, and output the allocation with the highest value.

We now show the existence of one ``good'' $t$ as above. Let $(O_1,\ldots O_n)$ be some optimal solution. Put the players into $\log m$ bins, where player $i$ is in bin $r$ if $2r>|O_i|\geq r$, for $r\in \{1,2,4, \ldots, m\}$. Let $N_r$ be the set of players in bin $r$. Let $t=\arg\max_r\Sigma_{i\in N_r}v_i(O_i)$. For each $i$ let $A_i$ be an arbitrary subset of $O_i$ of size $t$ if $i\in N_t$ and $A_i=\emptyset$ otherwise. We have that $\Sigma_iv_i(A_i) \geq \frac{\Sigma_{i\in N_t}v_i(O_i)}{2}\geq \frac {\Sigma_iv_i(O_i)} {4\log m}$.
\end{proof}

\end{proof}

In the next two subsections we to design two algorithms for $t$-restricted instances. We will use the proposition to claim that these algorithms can be extended with a small loss in the approximation factor to the general case as well.

\subsection{A Simultaneous $\tilde O(m^{\frac 1 3})$-Approximation }\label{sec-alg}

We show that simultaneous algorithms can achieve better approximation ratios than those that can be obtained by sketching the valuations. Specifically, we prove that:

\begin{theorem}\label{thm-xos-alg}
There is a deterministic simultaneous algorithm for combinatorial auctions with subadditive bidders where each player sends $poly(m,n)$ bits that guarantees an approximation ratio of $\tilde O(m^{\frac 1 3})$.
\end{theorem}

Given Proposition \ref{prop-restricted-instances}, we may focus only on designing algorithms for $t$-restricted instances. The algorithm for $t$-restricted instances is simple:
\begin{enumerate}
\item Each player reports a maximal set of disjoint bundles $\mathcal{S}_{i}$ such that for every bundle $S\in \mathcal{S}_{i}: |S|=\frac{t}{2}$ and $v_i(S)=|S|$.
\item For each $i$, let $v'_i$ be the following XOS valuation: $v'_i(S)=\max_{T\in \mathcal S_i}|T\cap S|$.
\item Output $(T_1,\ldots, T_n)$ -- the best allocation\footnote{As stated, this algorithm uses polynomial communication but may not run in polynomial time since finding the optimal solution with explicitly given XOS valuations is NP hard. If one requires polynomial communication \emph{and} time, an approximate solution may be computed using any of the known constant ratio approximation algorithms. The analysis remains essentially the same, with a constant factor loss in the approximation ratio.} with respect to the $v'_i$'s.
\end{enumerate}
Notice that the size of the message that each player sends is $poly(m)$. Furthermore, for each bundle $S$ and bidder $i$ $v_i(S)\geq v'_i(S)$. We will show that the best allocation with respect to the $v'_i$'s provides a good approximation with respect to the original valuations $v_i$'s. I.e., $\frac{\sum_i v_i(A_i)}{\sum_i v'_i(T_i)} ={\tilde O(m^{\frac 1 3})}$.

The proof considers three different allocations and shows that each allocation provides a good approximation for a different regime of parameters. 
\begin{itemize}
\item The best allocation (with respect to the $ v'_i$'s) in which each player receives at most one item. We show that this provides an $O(t)$ 
approximation with respect to the $v_i$'s.
\item Each player $i$ is allocated the fraction of the bundle $T\in \mathcal S_i$ that maximizes $| T\cap A_i|$. We show that this allocation guarantees an approximation ratio of $\tilde O(l)$, for some $l$ related to the $l_i$'s. 
\item The third allocation is constructed randomly (even though our algorithm is deterministic): each player $i$ chooses at random a bundle $\mathcal S_i$, and each item $j$ is allocated to some player that $j$ is in his randomly selected bundle, if such exists. Let $n'$ be the number of nonempty bundles in $(A_1,\ldots, A_n)$. We show that the expected approximation of this allocation is $\tilde O(n'/l)$, for the same $l$ as above. 
\end{itemize}
More formally:

\begin{lemma}\label{lemma-xos-alg-proof}
Let $(A_1,\ldots, A_n)$ be the allocation that is guaranteed by the $t$-restrictness of the instance. Let $(T_1,\ldots,T_n)$ be the allocation that the algorithm outputs. Then, $\frac{\sum_i v_i(A_i)}{\sum_i v_i(T_i)} \leq \tilde O(m^{1/3})$.
\end{lemma}
\begin{proof}
We first show that every player $i$ reports a large fraction of $A_i$ (but the  items of $A_i$ might be split among different reported bundles).  
\begin{claim}\label{claim-alg-t/2}
For each player $i$, let $\Phi_i=A_i\cap (\cup_{T\in \mathcal S_i}T)$. We have that $|\Phi_i|\geq t/2$.
\end{claim}
\begin{proof}
Suppose that there exists some player $i$ such that $|\Phi_i|<t/2$, and let $S=A_i-(\cup_{T\in\mathcal S_i}T)$. Since $|A_i|= t$, we have that $|S|\geq t/2$. But then the set $\mathcal S_i$ is not maximal, since any subset of the bundle $S$ of size of $t/2$ could have been added to it.
\end{proof}

\begin{corollary}
$\Sigma_iv_i(\Phi_i) \geq  \frac {\Sigma_i v(A_i)} {2}$.
\end{corollary}
We now divide the players into $\log m$ bins so that player $i$ is in bin $r$ if $r>|\mathcal S_i|\geq \frac r 2$. Let the set of players in bin $r$ be $N_r$. Observe that there has to be some bin $l$ with $\Sigma_{i\in N_l}v_i(\Phi_i)\geq \frac {\Sigma_iv_i(\Phi_i)} {\log m}$. Denote its size by $n_l=|N_l|\leq n'$. For the rest of this proof we will only consider the players in $N_l$; this will result of a loss of at most $\log(m)$ in the approximation factor. 

\begin{claim}[$\tilde O(t)$-approximation]\label{claim-approx-t}
$\Sigma_{i\in N_l} v'_i(T_i)\geq \Sigma_i v_i(\Phi_i)/(t \cdot \log(m))$.
\end{claim}
\begin{proof}
We show that there exists an allocation $(B_1,\ldots, B_n)$, $|B_i|\leq 1$ for all $i$, with $\Sigma_iv'_i(B_i)\geq 2\Sigma_iv_i(\Phi_i)/t$. For each player $i\in N_l$, let $B_i$ be the set that contains exactly one item $j_i$ from $\Phi_i$ if $\Phi_i\neq\emptyset$ and $B_i=\emptyset$ otherwise. Notice that the $j_i's$ are distinct, since for every $i\neq i'$ we have that $\Phi_i\cap \Phi_{i'}=\emptyset$. The lemma follows since $v'_i(\{j_i\})=1$.
\end{proof}

\begin{claim}[$\tilde O(l)$-approximation]\label{claim-approx-l}
$\Sigma_{i\in N_l} v'_i(T_i)\geq \Sigma_iv_i(\Phi_i)/( l\cdot \log m)$.
\end{claim}
\begin{proof}
For each player $i\in N_l$, let $B_i \subseteq \Phi_i$ be the bundle of the maximal size such that $v'_i(B_i)=|B_i|$. Note that $|B_i|\geq |\Phi_i|/l$ as by construction $v'_i$ has at most $l$ clauses and by definition each item in $\Phi_i$ is contained in one of the $|\mathcal S_i|$ clauses of $v'_i$. Therefore, $|B_i|\geq |\Phi_i|/l$. Thus we have that $\Sigma_{i\in N_l} v'_i(B_i)\geq \frac {\Sigma_{i\in N_l}v_i(\Phi_i)} {l}$. The claim follows as we already observed that $\Sigma_{i\in N_l}v_i(\Phi_i)\geq \frac {\Sigma_iv_i(\Phi_i)} {\log m}$. 
\end{proof}

\begin{claim}[$\tilde O(n'/l)$-approximation]\label{claim-approx-n/l}
$\Sigma_{i\in N_l} v'_i(T_i)\geq \frac{4n_l}{l} \cdot \frac{\Sigma_{i} v_i(\Phi_i)}{\log(m)}$.
\end{claim}
\begin{proof}
Consider the following experiment. For each player in $N_l$ choose uniformly at random a bundle to compete on among the bundles in the set $\mathcal S_i$ (recall that $\frac{l}{2}  \leq |\mathcal S_i|< l$). Now allocate each item $j$ uniformly at random to one of the players that are competing on bundles that contain that item $j$.

There are $n_l$ players in this experiment, so there are at most $n_l$ potential competitors for each item. Each player has at least $l/2$ disjoint potential bundles to compete on, hence the expected number of competitors on each item is at most $2n_l/l$. So when player $i$ is competing on a bundle, the expected competition on each of the items he competes on is $2n_l/l+1$. Therefore, the expected contribution of each item that player $i$ is competing on is at least $\frac {l} {2n_l}$. The lemma follows by applying linearity of expectation on the (at least) $t/2$ items that are in the bundle that player $i$ is competing on. Therefore we have that  $\Sigma_{i\in N_l} v'_i(T_i)\geq \frac {l} {2n_l} \cdot \frac{t}{2} \cdot {n_l} = \frac{l\cdot t}{4}$. Recall that $\frac{\Sigma_{i} v_i(\Phi_i)}{\log(m)} \leq \Sigma_{i\in N_l} v_i(\Phi_i) \leq t \cdot n_l $. Thus we have that $\frac{l\cdot t}{4} \geq \frac{4n_l}{l} \cdot \frac{\Sigma_{i} v_i(\Phi_i)}{\log(m)}$ as required. 
\end{proof}

By choosing the best of these three allocations we get the desired approximation ratio:

\begin{claim}\label{lemma-xos-alg-combined}
Suppose that we have three allocations $B^1,B^2$ and $B^3$ such that: $\frac{\sum_i v_i(A_i)}{\sum_i v'_i(B^1_i)} ={O(t)}$, $\frac{\sum_i v_i(A_i)}{\sum_i v'_i(B^2_i)} ={\tilde O(l)}$ and $\frac{\sum_i v_i(A_i)}{\sum_i v'_i(B^3_i)} =\tilde O(n'/l)$. Let $B$ be the allocation with the highest welfare among the three. Then, $\frac{\sum_i v_i(A_i)}{\sum_i v'_i(B_i)} ={\tilde O(m^{\frac 1 3})}$.
\end{claim}

\begin{proof}
By 
the first two claims,
we get an approximation ratio of $\tilde O(m^{\frac 1 3})$ whenever $l<m^{\frac 1 3}$ or $t<m^{\frac 1 3}$. Hence, we now assume that $l,t\geq m^{\frac 1 3}$. Now observe that when $t\geq m^{\frac 1 3}$ then $n'\leq m^{\frac 2 3}$, since there can be at most $m/t$ players that receive non-empty (disjoint) bundles in any allocation where the size of each non-empty bundle is 
at least $t$. We therefore have that $\frac {n'} l\leq m^{\frac 1 3}$ and the lemma follows by 
the third claim.
\end{proof}

\end{proof}

\subsection{A $k$-Round Algorithm}
We now develop an algorithm that guarantees an approximation ratio of $O(m^{\frac 1 {k+1}})$ for combinatorial auctions with subadditive valuations in $k$ rounds. In each of the rounds each player sends $poly(m)$ bits. We provide an algorithm for $t$-restricted instances (see Section \ref{sec-alg} for a definition). By Proposition \ref{prop-restricted-instances} this implies an algorithm with almost the same approximation ratio for general subadditive valuations.

\subsubsection*{The Algorithm (for $t$-restricted instances)}

\begin{enumerate}
\item Let $N_1=N$, $U_1=M$ and $U_{1,i} = M$.
\item In every round $r=1,\ldots,k$:
\begin{enumerate}
\item Each player reports the a maximal set of disjoint bundles $\mathcal{S}_{r,i}$ such that for every bundle $S\in \mathcal{S}_{r,i}: S\subseteq U_{r,i},~|S|=\frac{t}{2k}$ and $v_i(S)=|S|$.
\item Go over the players in $N_r$ in an arbitrary order. For every player $i$ for which there exists a bundle $S \in \mathcal{S}_{r,i}$ such that at least $\frac{1}{m^{\frac{1}{k+1}}}$ of its items were not allocated yet, allocate player $i$ the remaining unallocated items of $S$.
\item Let $N_{r+1} \subseteq N_r$ be the set of players that were not allocated items at round $r$ or before.
\item Let $U_{r+1} \subseteq U_r$ be the set of items that were not allocated at round $r$ or before.
\item Let $U_{r+1,i} = (\cup_{S \in \mathcal{S}_{r,i}} S) \cap U_{r+1}$.
\end{enumerate}
\end{enumerate}
\begin{theorem}\label{thm-xos-rounds}
For every $k\leq \log m$, there exists an algorithm for $t$-restricted instances that provides an approximation ratio of $O(k\cdot m^{\frac 1 {k+1}})$ in $k$ rounds where each player sends $poly(m,n)$ bits. In particular, when $k=O(\log m)$ the approximation ratio is $O(\log m)$. As a corollary, there exists a $k$-round approximation algorithm for subadditive valuations that provides an approximation ratio of $O(k\cdot m^{\frac 1 {k+1}}\cdot \log^3 m)$.
\end{theorem}
\begin{proof}
In the analysis we fix some optimal solution $(O_1,\ldots, O_n)$. We break the proof of the theorem into two lemmas.

\begin{lemma}\label{claim-xos-rounds-ends}
At the end of the algorithm, either for every player $i$ that did not receive any bundle we have that $U_{k+1,i}=\emptyset$, or the approximation ratio of the algorithm is $O(m^{\frac 1 {k+1}})$.
\end{lemma}
\begin{proof}
Observe the for every player $i$ and round $r$, if $i \in N_r \cap N_{r+1}$ then $|\cup_{S \in \mathcal{S}_{r,i}} S| \geq m^{\frac{1}{k+1}} |\cup_{S \in \mathcal{S}_{r+1,i}} S|$. This is true since if player $i$ was not allocated any items at round $r$ then for every bundle $S \in \mathcal{S}_{r,i}$ that he reported, at least $(1-\frac{1}{m^{\frac{1}{k+1}}})$ of the items were allocated. Thus, $m^{\frac{1}{k+1}} \cdot|U_{r+1,i}| \leq |\cup_{S \in \mathcal{S}_{r,i}} S|$. Recall that by definition we have that $(\cup_{S \in \mathcal{S}_{r+1,i}} S) \subseteq U_{r+1,i}$. Therefore, for every player $i$ that was not allocated anything at round $k$, we have that $|U_{k+1,i}| \leq m^{\frac{1}{k+1}}$. If there exists a player $i$ such that $|U_{k+1,i}| >0$ then it has to be the case that initially $|\cup_{S \in \mathcal{S}_{1,i}} S|\geq m^{\frac k {k+1}}$. This implies that at least $m^{\frac k {k+1}}-m^{\frac{1}{k+1}}$ of the items were allocated and hence the approximation ratio is $O(m^{\frac 1 {k+1}})$ (there are $m$ items, so the value of the optimal solution is at most $m$). In any other case for every player that was not allocated any bundle we have that $U_{k+1,i}=\emptyset$.
\end{proof}

\begin{lemma}
Suppose that for every player $i$ that did not receive any bundle we have that $U_{k+1,i}=\emptyset$. The approximation ratio of the algorithm is $O(k\cdot m^{\frac 1 {k+1}})$.
\end{lemma}
\begin{proof}
Suppose that at least $n'/2$ players were allocated non empty bundles. The value of the bundle that each player was allocated is at least $\frac{t}{2k \cdot m^{\frac 1 {k+1}}}$. Thus, in this case the algorithm achieves an $O(k\cdot m^{\frac 1 {k+1}})$-approximation (recall that the value of the optimal solution is $n'\cdot t$). Else, assume that at most $n'/2$ players were allocated some bundle. Consider player $i$ in $N'$ that did not receive any items (there are at least $n'/2$ such players). We will show that for each such player $i$, $|O_i \cap U_{k}| \leq t/4$. In other words, we will show that for every such player $i$, the algorithm allocates at least $\frac {3t} {4}$ of the items that this player receives in the optimal solution. This will prove an approximation ratio of $O(k\cdot m^{\frac 1 {k+1}})$. 

Thus assume towards contradiction that there exists some $i\in N$ that did not receive any items for which $|O_i \cap U_{k}| \geq 3t/4$. In this case by Claim \ref{clm:ind} (below) we have that $|O_i \cap U_{k,i}| \geq \frac{3t}{4}-  k \cdot \frac{t}{2k} = \frac{t}{4}$. In particular this implies that $|U_{k,i}| \geq t/4$. This is a contradiction to the assertion of Lemma \ref{claim-xos-rounds-ends} that $U_{k,i} = \emptyset$ for every player that was not allocated any bundle.
\end{proof}

\begin{claim} \label{clm:ind}
For every $i\in N_r\cap N'$, $|O_i \cap U_r| - |O_i \cap (\cup_{S \in \mathcal{S}_{r,i}} S)| \leq r \cdot \frac{t}{2 k}$. 
\end{claim}
\begin{proof}
We prove the claim by induction. For the base case $r=1$, observe that since player $i$ reports a maximal set of bundles of value $\frac{t}{2k}$ at most $\frac{t}{2k}$ items of his optimal bundle may not be reported by him -- otherwise, those items can be bundled together and added to $\mathcal S_{r,i}$, contradicting maximality. We now assume correctness for $r$ and prove for $r+1$. We have that $|O_i \cap U_r| - |O_i \cap (\cup_{S \in \mathcal{S}_{r,i}} S)| \leq r \cdot \frac{t}{2 k}$. Denote the items allocated at round $r$ by $A_r$. We have that $|O_i \cap (U_r \setminus A_r)| - |O_i \cap (\cup_{S \in \mathcal{S}_{r,i}} S) \setminus A_r| \leq r \cdot \frac{t}{2 k}$ since $\cup_{S \in \mathcal{S}_{r,i}} S) \subseteq u_r$. By definition, this implies that $|O_i \cap U_{r+1}| - |O_i \cap u_{r+1,i})| \leq r \cdot \frac{t}{2 k}$. The proof is completed by observing that $|O_i \cap U_{r,i}|-|O_i \cap (\cup_{S \in \mathcal{S}_{r,i}} S)| \leq \frac{t}{2 k}$. Since in every round player $i$ reports a maximal disjoint set of bundles of size $\frac{t}{2 k}$, thus at every round the player can leave out of its reported set at most $\frac{t}{2 k}$ items from the optimal solution (i.e., if the player leaves out more items then he could have reported those items, contradicting maximality).
\end{proof} 

\end{proof}

\bibliographystyle{plain}

\bibliography{bib}

\begin{appendix}

%

\section{Communication Complexity of Bipartite Matching}\label{appendix-matching}

In this appendix we discuss a proposed
communication complexity investigation of the bipartite matching
problem.  This model is essentially the same as that used
in our investigation of bipartite matching in
the rest of the paper but focusing on the exact problem rather than
on approximations, and proposing the study of communication
as of itself rather than merely as an abstraction of market processes.

We focus on the open problem(s) and shortly mention some 
related models where communication 
bottlenecks for matching have been investigated and give a few pointers
to different such threads where the interested reader may find many more references.

\subsection{The Model and Problem}

There are $n$ items and $n$ players.  
Each player $i$ holds a subset $S_i$ of items
that he is interested in.  I.e. we have a bipartite graph with $n$ left vertices (players)
and $n$ right vertices (items), and have a player in our model for each left vertex, a player that knows the set of neighbors of his vertex in the bipartite graph (but 
there are no players associated with the items.)  The goal of these players is to 
find a maximum matching between items and players, i.e. that each player is assigned
a single item $j_i \in S_i$ with no items assigned to multiple players $j_i \ne j_{i'}$
for $i \ne i'$.

\subsubsection*{Communication Model}

The players engage in a fixed communication protocol using broadcast messages.  Formally
they take turns writing on a common ``blackboard''.  At every step
in the protocol, the identity of the next player $i$ to
write a bit on the blackboard
must be completely determined by the contents of
the blackboard, and the message written by this player $i$ must be
determined by the contents
of the blackboard as well as his own private input $S_i$.  Whether
the protocol terminates at a given point must be completely determined by
the contents of the blackboard, and at this point the output matching must
be solely determined by the contents of the blackboard.  This model
is completely equivalent to a decision tree, each query 
can be an arbitrary function depending only on a single player's information $S_i$.
The measure of complexity here is the total number of bits communicated.

\subsubsection*{Rounds}

The communication model above allows an arbitrary order of communication.  Of interest are
also various limited orders: oblivious (the order of speaking is fixed independently of
the input), and the simplest special case of it, 
one-way communication where the players speak in the fixed order of player 1, player 2, etc.
We will focus on speaking in rounds: in each ``round'' each of the $n$ players writes
a message on the blackboard, a message that may depend on his own input as well as 
the messages of the others in previous rounds (i.e. on the contents
of the blackboard after the previous round).  The measures of complexity here are
the
number of rounds and the total number of bits communicated.  The special
case of a single round is called a simultaneous protocol.

\subsubsection*{Open Problems}

How much communication is needed for finding a maximum matching?  What if 
we are limited to $r$ rounds?  These questions apply both to deterministic and to
randomized protocols.  The tradeoff between communication and approximation
is of course also natural to explore.

\subsubsection*{What is Known}

The trivial upper bound for communication is $n^2$ since players can all simultaneously
send their full input.  The trivial lower bound is $\Omega(n\log n)$ as this is the number
of bits necessary to represent the output matching (and every matching may need to be 
given as output).  

Significantly, the non-deterministic (and co-non-deterministic) 
communication complexity is
also $O(n \log n)$: to verify that a given matching is maximum size it suffices to add
a Hall-theorem blocking set, or alternatively a solution for the dual.
Specifically, a specification of a set of ``high-price'' items, 
so that (1) only allocated items are high-price (2)
all players that are not allocated a low price item 
are only interested in high-price ones.  The fact that the non-deterministic complexity
is low means that ``easy'' lower bounds techniques
such as fooling-sets or cover-size bounds
will not suffice for giving good lower bounds.

Interestingly, an $O(n^{1.5}\log n)$ upper bound can be obtained by adapting
known algorithms to this framework:  First, the auction algorithm described in
Section \ref{sec-matching-alg} gives a $(1-\delta)$-approximation using $O(n\log n / \delta)$ 
communication. When we choose $\delta = 1/\sqrt{n}$ this means that we get
a matching that is at most smaller than the optimal one by an additive $\sqrt{n}$.
We can thus perform $\sqrt{n}$ more augmenting path calculations to get an optimal
matching.  Each augmenting path calculation requires only $O(n \log n)$ bits 
of communication: it requires finding a path in a graph 
on the players that has a directed edge
between player $i$ and player $i'$ whenever $i$ is interested in
an item that is currently allocated to $i'$. The goal here is to find a path from any 
player that is not allocated an item to any player that is interested in an unallocated
item.  A breadth first search with the blackboard serving as the queue requires 
writing every vertex at most once on the blackboard, at most $O(n\log n)$ communication.

We do not know any better upper bound, nor do we know a better than $O(n^2)$ upper
bound for even $n$ rounds.  Our lower bounds for matching provide an 
$\Omega(n^2)$ lower bound for simultaneous protocols, and a $n^{1+\Omega(1/\log\log n)}$
lower bound for one-way communication follows from \cite{GKK12}.
We don't know any lower bound better than $\Omega(n\log n)$ for general protocols or
even for 2-round protocols.

\subsubsection*{Algorithmic Implications}

We believe that studying the bipartite matching under this model may be a productive way of understanding the general
algorithmic complexity of the problem.  A major open problem is whether bipartite matching has a (nearly-)linear time 
algorithm: $O(n^{2+o(1)})$ time for dense graphs (and maybe $O(m^{1+o(1)})$ for graphs with $m$ edges).  The best 
deterministic running time known (for the dense case) is the 40-year old $O(n^{2.5})$ 
algorithm of \cite{HK73}, with a somewhat better
randomized $O(n^\omega)$ algorithm known \cite{MS04} (where $\omega=2.3...$ is the matrix multiplication exponent).  For special cases
like regular or near-regular graphs nearly linear times are known (e.g. \cite{Y13}).  In parallel computation, a major
open problem is whether bipartite matching  can be solved in parallel poly-logarithmic time
(with a polynomial amount of processors).  Randomized parallel algorithms for the problem \cite{MVV87, KUW85} have been 
known for over 25 years.

On the positive side, it is ``likely'' that any communication protocol for bipartite matching that improves on the currently known $O(n^{1.5})$
complexity will imply a faster than the currently known $O(n^{2.5})$ algorithm.  This is not a theorem, however
the computational complexity needed to send a single bit in a communication protocol is rarely more than linear in the input held by the player sending the 
bit.  Most often each bit is given by a very simple computation in which case this is so trivially, but sometimes, clever data structures will be needed
for this to be so.  If this will be the case in the communication protocol in question then the improved algorithm is implied.  A similar phenomena 
should happen if a deterministic communication protocol that uses poly-logarithmic many rounds is found since most likely 
each bit sent by each player is determined by a simple computation that
can be computed in parallel logarithmic time.

On the negative side, lower bounds in a communication model do not imply algorithmic lower bounds, however they can direct the search for algorithms by
highlighting which approaches cannot work.  This is the great strength of concrete models of computation where lower bounds are possible to prove.

\subsection{Related Models}

The bipartite matching problem has been studied in various models that focus on communication. 

\subsubsection*{Point to Point Communication} In our model players communicate using a ``blackboard''; any bits sent by a player is seen by all.
A weaker model that is more natural to capture realities in distributed systems will consider the case 
where each message is sent to a single recipient. Such models are also called ``message passing'' or ``private channels''.
One must be slightly careful in defining such protocols as to ensure that
no communication is ``smuggled'' by the timing of messages, and the standard way of doing so is using the essentially equivalent 
coordinator model of \cite{DF89}.  It turns out that bipartite matching is even harder to approximate without broadcast and the
results of
\cite{HBMQ13} give an $\Omega(\alpha^2 n^2)$ lower bound for even finding an $\alpha$-approximation (even using randomization).  
Note that this model is trivially
more general than that of simultaneous protocols hence this lower bound gives the randomized lower bound from Theorem \ref{theorem-matching-rand-lb}.
We also note that our proof for this theorem actually also applies to the model of interactive protocols with private channels.  

\subsubsection*{Multiple Vertices per Processor} In our model, the problem for $n+n$-vertex graphs is 
handled by $n$ players.
A model that is more appropriate for the current scales of distributed systems is to use
$k$ processors where $k<<n$.  There are
various options for partitioning the $n^2$ bits of input 
to the processors where, most generally,
each of the $k$ processors can hold an arbitrary part.  This is the usual model of interest for distributed systems and was e.g. used in \cite{HBMQ13}. 
Here again one might distinguish between broadcast and point to point communication models, where the gap between the models
can be no larger than a factor of $k$.  To convert an auction-based algorithm to run in this framework one must be able
to calculate the demand of a vertex.  Usually this can be easily done with $O(k \log n)$ communication, but in many cases
it is possible to pay an overhead of $O(\delta^{-1})$ instead.  In particular, the basic $\delta$-approximation auction algorithm
that obtains a $(1-\delta)$ approximation 
can be run in this model using $O(n\log n/\delta^2)$ communication in the blackboard model and thus $O(kn\log/\delta^2)$ with point to point communication.

\subsubsection*{Two Players} When we are down to $k=2$ players we are  back to the standard two-player communication complexity model of Yao.  Two variants
regarding the partition of the input to the two players are natural here: (a) each player $i$ holds an arbitrary subset $E_i$ of the edges
and the graph in question is just the union $E_1 \cup E_2$; (b) each player $i$ holds the edges adjacent to $n/2$ of the left-vertices.  These models
are less-distributed and thus stronger than our $n$-player model.  As in our model, the complexity of bipartite matching is completely open,
and in particular the communication 
complexity of
the decision problem of whether the input graph has a perfect matching is open with
no known non-trivial, $\omega(n)$, lower bounds or non-trivial, $o(n^2)$
upper bounds.  

\subsubsection*{Streaming and Semi-streaming} One of the main applications of communication complexity is to serve as lower
bounds for ``streaming'' algorithms, those are algorithms that go over the input sequentially in a single pass (or in few passes), while using only a modest amount
of space.  The model of communication complexity required for such lower bounds is that of a one-way single-round private-channel protocol where 
in step $i$ player $i$ sends a message to player $i+1$.  (For $r$-pass variants of streaming algorithms, we will have $r$ such rounds of one-way communication.)  
The lower bounds mentioned above thus imply that no streaming algorithm that uses $o(n)$ space can get even a constant factor-approximation
of the maximum matching, even with  $O(1)$ rounds.  A greedy algorithm gets $1/2$-approximation in a single round using $O(n)$ space,
and slight improvements in the approximation factor using linear space are possible, e.g. 
using the online matching algorithm of \cite{KVV90}.  In $r$ passes and nearly-linear space, \cite{K12} gets an $1-O(1/\sqrt{r})$ approximation.  
Streaming algorithms that use linear or near-linear space are usually called semi-streaming algorithms and lower bounds for them are
usually derived by looking at the information transfer between the ``first half'' and the ``second half'' of the input data and 
proving a significantly super-linear lower bound on the one-way two-party communication. This
was done in \cite{GKK12} who give a $n^{1+\Omega(1/\log\log n)}$ lower bound for improving the $2/3$ approximation.  To the best of our knowledge
no better lower bound is known even for getting an exact maximum matching.

\subsubsection*{Distributed Computing} In this model the input graph is also the communication network.  I.e. players can communicate with each other only over links
that are edges in the input graph, and the interest is the number of rounds needed.
For this to make sense in a bipartite graph we need to also have processors for the right-vertices of the graph (and thus every edge is known by the
two processors it connects.)  
It is not hard to see that to get a perfectly maximal matching in this model 
transfer of information across the whole diameter of the graph may be needed, which may require $\Omega(n)$
rounds of communication, but in 
\cite{LPP08} a protocol is exhibited that gives a $(1-\delta)$-approximation in $O(\log n / \delta^3)$ rounds.



\section{Gross Substitutes Valuations}\label{app-gs}

\begin{proposition}
Any exact simultaneous algorithm for combinatorial auctions with gross substitutes valuations (even for just two players) requires super polynomial message length.
\end{proposition}
\begin{proof}
Suppose towards a contradiction that there exists an exact algorithm $A$ for combinatorial auctions with gross substitutes where the message length of each message is polynomial. We will show how given such an algorithm $A$ we can construct an exact sketch any gross substitutes valuation using only polynomial space. We get a contradiction since gross substitutes valuations do not admit polynomial sketches \cite{BH11}. Since the construction of \cite{BH11} contains only integer-valued functions between $1$ and $poly(m)$ it suffices to show how to sketch these functions only.

Given such an algorithm $A$ the sketch for a valuation $v$ we construct is simply the message $L_v$ a player with valuation $v$ sends in the algorithm $A$ and in addition the value $v(M)$. We now show how to compute the value of any bundle $v(S)$ with no additional communication cost. For this it suffices to show that for any $T \subset M$ and any item $j\notin T$ we can compute the marginal value $v(j|T)$. 

To compute the marginal value $v(j|T)$, we construct the following family of additive valuations:

$$v^{j|T}_x(\{j'\}) = \begin{cases}

  x & \text{for $j'=j$} \\

  0 & \text{for $j' \in T$} \\

  2v(M) & \text{otherwise}

\end{cases}$$

Consider an instance with two players, one with valuation $v$ and the other with valuation $v^{j|T}_x$. Observe that in any efficient allocation of this instance  the player with valuation $v$ should always receive bundle $T$. He will also receive item $j$ if $v(j|T)>x$ and will not receive it if $v(j|T)<x$. Thus, we can use $A$ to compute the value of $v(j|T)>x$ (with no additional communication cost) by giving it as an input $L_v$ and $L_{v^{j|T}_x}$ for increasing values of $x \in \{\frac 1 2 , \frac 3 2, \ldots, v(M)-\frac{1}{2}\}$. This allows to determine the maximal value of $x$ ($x^*$) for which the $v$-player receives $j$. We now have that $v(j|T)= x^*+\frac{1}{2}$ as we assumed that the valuation functions only use integer values.
\end{proof}

%

\end{appendix}


\end{document}